\documentclass[journal]{IEEEtran}
\usepackage{bm}
\usepackage{xcolor}
\usepackage{graphicx}
\usepackage{mathptmx}
\usepackage{amsmath}
\usepackage{amssymb}
\usepackage{amsthm}
\usepackage[numbers, sort&compress]{natbib}

\usepackage{amsmath}
\usepackage{algpseudocode}
\usepackage{algorithm}
\usepackage{threeparttable}
\usepackage{listings}
\usepackage{float}
\usepackage{booktabs}
\usepackage{mathtools} 
\usepackage{cleveref }
\usepackage{makecell}

\newtheoremstyle{mytheorem}
  {3pt}
  {3pt}
  {\itshape}
  {}
  {\itshape}
  {.}
  {.5em}
  {\thmname{#1}\thmnumber{{ }#2}%
   \thmnote{ {\the\thm@notefont(#3)}}}
\makeatother
\theoremstyle{mytheorem}

\newtheorem{thm}{Theorem}[section]
\newtheorem{lem}[thm]{Lemma}

\newtheorem{cor}{Corollary}

\newtheorem{defn}{Definition}[section]

\newtheorem{rmk}{Remark}

\ifCLASSINFOpdf
\else
\fi
\begin{document}
%
\title{Waveform Design for Optimal PSL Under Spectral and Unimodular Constraints via Alternating Minimization}
%
%
%

\author{Chin-Wei Huang, Li-Fu Chen, Borching Su, \IEEEmembership{Member,~IEEE}
\thanks{}
\thanks{}
\thanks{}}

%
%

\markboth{}%
{Shell \MakeLowercase{\textit{et al.}}: Bare Demo of IEEEtran.cls for IEEE Journals}
%



\maketitle

\begin{abstract}
In an active sensing system, waveforms with good auto-correlations are preferred for accurate parameter estimation. Furthermore, spectral compatibility is required to avoid mutual interference between devices as the electromagnetic environment becomes increasingly crowded. 
Waveforms should also be unimodular due to hardware limits. 
In this paper, a new approach to generating a unimodular sequence with an approximately optimal peak side-lobe level (PSL) in auto-correlation and adjustable stopband attenuation is proposed. 
The proposed method is based on alternating minimization (AM) and numerical results suggest that it outperforms  existing methods in terms of PSL.
We also develop a theoretical lower bound for the PSL minimization problem under spectral constraints and unimodular constraints, which can be used for the evaluation of the results in various works about this waveform design problem. 
It is observed in the numerical results that the PSL of the proposed algorithm is close to the derived lower bound.

\end{abstract}


%
\IEEEpeerreviewmaketitle

\section{Introduction}\label{sec:intro}
\IEEEPARstart{I}{n} an active sensing system, such as radar or sonar, valuable properties of the targets can be determined by transmitting waveforms to an area of interest and analyzing the received signals \cite{he_li_stoica_2012}.
For example, given the propagation speed of radar waves, we can estimate the distance between the radar and the target by measuring the round-trip time delay. 
The target's speed can be calculated by measuring the Doppler frequency shift of the received signal \cite{skolnik2008radar}.
It comes as no surprise that a good design of the transmitted waveform can not only lead to accurate parameter estimation but also a reduced computational burden at the receiver \cite{he_li_stoica_2012, skolnik2008radar, He2009}.

Since the matched filter is commonly applied for range compression to maximize the signal-to-noise ratio \cite{he_li_stoica_2012, He2010}, a waveform whose auto-correlation exhibits low side-lobes is desirable \cite{He2009, levanon2004radar}. 
Furthermore, the unimodular property is generally desirable for a radar waveform because of the practical hardware restriction \cite{he_li_stoica_2012, Gri2015}.
In general, the auto-correlation quality can be quantified through two metrics: integrated side-lobe level (ISL) and peak side-lobe level (PSL). 
Numerous design methods of unimodular waveform for low ISL were proposed \cite{He2009, Mohammad2017, Song2016April, Song2016June}, such as the majorization-minimization (MM) method proposed by Song et al. \cite{Song2016April, Song2016June}, and the well-known Cyclic algorithm-new (CAN) proposed by He et al. \cite{He2009}, etc. 
In some early studies, sequences with low PSL were often designed in closed-form, such as the works on fixed-length polyphase Barker sequences or some families of polyphase sequences \cite{Frank1963, Golomb1965}. 
However, sequence designs that directly optimize the PSL has not been found until the recent years.
In  \cite{Song2016April}, Song et al. proposed a design algorithm that becomes the first to put the PSL in the objective function, taking advantage of the $\ell_{p}$-norm approximation of PSL.
Following this method, many variant methods have been proposed since then \cite{Song2016June, Hamid2017, Raei2022, Cui2018, FAN2021107960, LU2022}. 

Another important aspect of waveform design that receives increasing emphasis has been put on spectral compatibility since the proposal of cognitive radar. 
For cognitive radar, it is essential to adapt the spectrum of the transmitted waveform based on the changing environment \cite{Gri2015, Haykin2006, Michael2010, Aubry2014}. 
As a result, more and more researchers considered spectral suppression when designing waveform sequences \cite{He2010, Aubry2014, Aubry2020}. The ISL minimization with spectral constraints was also widely studied. 
In \cite{He2010}, the authors proposed the predominant stopband cyclic algorithm new (SCAN) algorithm. 
It minimizes the ``almost equivalent'' ISL metric along with the total stopband energy. 
In \cite{BISKIN2020102867}, the authors used MM method to minimize the ``almost equivalent'' ISL metric along with the stopband spectral energy. 
As for the PSL minimization with spectral constraints, the studies mostly started only in the recent years because of its extreme difficulty.
In \cite{Cui2018}, the authors proposed the frequency nulling modulation (FNM) to jointly minimize the $\ell_p$-norm approximated PSL and the stopband energy. 
In \cite{FAN2020107450}, the authors used the proximal method of multipliers (PMM) to minimize PSL with spectral constraints for multi-sequence design. 
In \cite{FAN2021107960}, the authors proposed the block successive upper-bound minimization (BSUM) technique to minimize the PSL with spectral constraints.
In \cite{LU2022}, the researchers combined the MM method with the PMM to include the spectral constraints in the local PSL minimization problem. 
All the above methods used either the MM method \cite{Cui2018, LU2022} or the ``almost equivalent'' property \cite{FAN2020107450, FAN2021107960} to address the quartic form in the optimization processes, and most of them (except for \cite{FAN2020107450}) rely on the $\ell_p$-norm approximation for dealing with the PSL optimization. 
However, the performance will be limited by the $p$ in the $\ell_p$-norm approximation due to the increasing computation for a better approximation \cite{LU2022}. 

{\tiny 
}

In this paper, a new method for PSL minimization problem for a unimodular sequence under the spectral constraint is proposed.
The main technique involved in the propose method is the transformation of the PSL minimization problem into a bi-convex problem, which can be solved by the exact penalty approach \cite{AM2013,Demir2014, Huang2022}. 
The proposed method, as the numerical results will suggest, outperforms existing methods by a considerable margin. 
The proposed method does not use the $\ell_p$-norm approximation, and the alternating minimization involved in the exact penalty approach allows us to tackle the quartic problem without resorting to the ``almost equivalent'' property or the MM method. 
In addition, our approach allows designers to limit the peak stopband energy with adjustable resolution.

The other important contribution of the paper is the derivation of a lower bound for the PSL minimization problem under the spectral constraint and the unimodular constraint. 
Such a lower bound sheds light on the analysis of the optimality gap and enables the evaluation the waveform performance by its distance between the attained PSL. 
In \cite{Welch_bound}, the well-known lower bound for PSL in multiple sequence design was proposed, but the corresponding lower bound for single sequence design is zero.
The largest lower bound for PSL in single sequence design, to the authors' best knowledge, is still the trivial one with value one \cite{McCormick2017}, not to mention the one with spectral constraints. 
In comparison, the proposed lower bound is much larger in the considered cases and even close to the PSL of our designed waveforms.


The rest of the paper is organized as follows. In Section \ref{sec:PSL_min}, we formulate the PSL minimization problem under the constraints and present the proposed algorithm. 
In Section \ref{Sec:Numerical_Examples}, numerical results are shown to demonstrate the advantages of the proposed methods. 
In Section \ref{sec:PerformanceAnalysis}, a lower bound of the main problem is derived followed by a performance analysis of the proposed algorithm.  
In Section \ref{sec:discussion}, some discussions about the proposed algorithm with the related works are made.
Conclusions are made in Section \ref{sec:conclusion}.

\subsection{Notations}
\label{Sec:Notations}
Boldface lowercase letters denote vectors, while boldface uppercase letters denote matrices. Operators $\left( \cdot \right)^H,\left( \cdot \right)^T$ and $\lVert  \cdot  \rVert_p$ denote the conjugate transpose, transpose and $p$-norm for matrices/vectors, respectively. 
We use $\mathrm{tr(\cdot)}$ to denote the trace of a matrix and $\left( \cdot \right)^*$ to denote the conjugate for a complex number.
For any positive integer $n$, $\mathbb{Z}_{n}$ stands for the set $\left\{ 0,1,\dots,n-1 \right\}$. 
The $n$-dimensional real and complex vector spaces are expressed as $\mathbb{R}^{n}$ and $\mathbb{C}^{n}$, respectively. 
The set of all $n\times n$ Hermitian matrices is denoted by $\mathbb{H}^n$, and the set of all positive semidefinite matrices is denoted by $\mathbb{H}^n_{+}$. 
For $\mathbf{A},\mathbf{B} \in \mathbb{H}^n$, the notation $\mathbf{A}\succeq\mathbf{B}$ means $\mathbf{A}-\mathbf{B}\in\mathbb{H}^n_+$. 
For $\mathbf{a},\mathbf{b}\in \mathbb{R}^n$, the notation $\mathbf{a}\succeq\mathbf{b}$ means all elements in the vector $\mathbf{a}-\mathbf{b}$ are non-negative. 
We adopt zero-based indexing throughout the paper. 
For a vector $\mathbf{x}$ and a matrix $\mathbf{X}$, the $i$-th entry of $\mathbf{x}$ and the $(i,j)$-th entry of $\mathbf{X}$ are denoted by $x_i$ and $\mathbf{X}(i,j)$, respectively. 
For convenience, all zero matrices and vectors are all denoted by $0$.

\section{PSL Minimization Problem}\label{sec:PSL_min}

\subsection{Problem Formulation}
Many works in the previous years have been focusing on the research about minimizing the integrated side-lobe (ISL) with unimodular constraints \cite{he_li_stoica_2012,He2010}. 
However, when the scenario such as active sensing with the threshold detection being applied for the target detection is considered, a sequence having a narrow side-lobe with a high level may cause a high false-alarm rate while leading to a high PSL but low ISL. Thus, in this case, the detection performance is dictated by the PSL instead of ISL \cite{Hamid2017}.
Moreover, since the spectral regulations for communications usually adopt the spectral masks \cite{3gpp.38.101-1} to define the acceptable spectral leakage, what we care about will be the peak energy instead of the total energy of the spectral leakage. Therefore, for the compatibility of communication applications, it is more sensible to design the waveform with the constraint on its maximal energy in the stopband. 

From the above reasons, we aim to design a unimodular sequence with low PSL and the constraint on the maximal energy in stopband. Before formulating the design problem into an optimization problem, we need to define the autocorrelation for the explicit expression of PSL first. 
\begin{defn}
Given a sequence $\mathbf{x} \in \mathbb{C}^N$, the aperiodic auto-correlation of $\mathbf{x}$ is defined as 
$$
r_\ell=\sum_{n=\ell}^{N-1}x_nx^*_{n-\ell},~\ell=0,\dots,N-1.
$$
\end{defn}
\begin{rmk}
It is easy to verify that $r_\ell = r_{-\ell}^*$ for all $\ell$.
In particular, $r_0$ is real-valued and is called the in-phase correlation.
All the other $r_\ell$'s are called the correlation side-lobes.
\end{rmk}
From the above definition, the peak side-lobe level (PSL) is defined as
\begin{equation}\label{eq:psl}
    \operatorname{max}\{|r_\ell|\}_{\ell=1}^{N-1}.
\end{equation}
Then, the optimization problem is formulated as follows:
\begin{subequations}\label{prob:original}
\begin{align}
& \underset{\mathbf{x}\in\mathbb{C}^{N}}{\mathrm{minimize}} && \operatorname{max}\{|r_\ell|\}_{\ell=1}^{N-1} \label{prob:originala}\\
& \text{subject to} && |X(f)|^2\leq U_{\mathrm{max}}, \forall f \in \mathcal{F}_{\mathrm{stop}}\label{prob:originalb}\\
&&& |x_n|=1, \forall n\in \mathbb{Z}_{N},\label{prob:originalc}
\end{align}
\end{subequations}
where $\mathbf{x}$ is the desired sequence, the objective function in \eqref{prob:originala} is the PSL, $X(f)$ is the discrete-time Fourier transform of the desired sequence at normalized frequency $f$, the constraint \eqref{prob:originalb} is the constraint of spectral compatibility and the constraint \eqref{prob:originalc} is the constraint of constant modulus.

However, the optimization problem \eqref{prob:original} is not convex and some reformulation is necessary. 
We propose to reformulate the problem through the technique of semi-definite relaxation (SDR).
Therefore, in the rest of this section, we will reformulate the objective function and the constraints in the problem \eqref{prob:original} one by one in Section \ref{sec:correlation}, \ref{sec:spectral_compatibility}, \ref{sec:unimodular}, and give a further transformation of the problem \eqref{prob:original} in Section \ref{sec:reformulation} for the proposed algorithm in Section \ref{sec:optimization_algorithm}.

\subsubsection{Auto-Correlation}\label{sec:correlation}
We first define the following nilpotent matrix to turn the auto-correlation into a quadratic form, which may be easier to address in an optimization problem.
\begin{defn}
For any positive integer $N$, the upper shift matrix  $\mathbf{N}_N$ is an $N \times N$ matrix defined by 
$$
\mathbf{N}_N= \begin{bmatrix}0 & 1 & 0 &\cdots & 0\\
0 & 0 & 1 & \ddots & \vdots\\
\vdots & & \ddots & \ddots & 0 \\
\vdots & &  & \ddots & 1 \\
0& \cdots & \cdots & \cdots & 0
\end{bmatrix}.
$$
\end{defn}
With the upper shift matrix $ \mathbf{N}_N$, we can express $r_\ell$ as
\begin{equation}\label{cst:psl_matrix}
    r_\ell= \sum_{n=\ell}^{N-1}x^{*}_{n-\ell}x_n =\mathbf{x}^H\mathbf{N}_N^\ell\mathbf{x},
\end{equation}
for $\ell \in \mathbb{Z}_N$.
Note that the quadratic form in \eqref{cst:psl_matrix} is of complex value because any positive integer power lower than $N$ of the upper shift matrix $\mathbf{N}_N$ is non-Hermitian. 

\subsubsection{Spectral Compatibility}\label{sec:spectral_compatibility}
Since a quadratic form may be easier to address than the square of an absolute value in an optimization problem, we first define the following vector for the formulation of the quadratic form.
\begin{defn}
Given a real number $f$, the vector $\mathbf{f}(f)$ is defined as
\begin{equation*}
\mathbf{f}(f)=[1~e^{j(2\pi f)\times 1}~\cdots~e^{j(2\pi f)\times(N-1)}]^T.
\end{equation*}
In addition, we define the matrix $\mathbf{F}(f)$ associated with $\mathbf{f}(f)$ to be
\begin{equation*}
\mathbf{F}(f)=\mathbf{f}(f)\mathbf{f}(f)^H.
\end{equation*}
\end{defn}
Without loss of generality, we only consider the normalized frequencies from $0$ to $1$. Given a subset $\mathcal{F}_{\mathrm{\text{stop}}}\subseteq \left[0,1\right]$, we aim to control the energy spectral density of $\mathbf{x}\in \mathbb{C}^N$ over $\mathcal{F}_{\mathrm{\text{stop}}}$. That is, we want to suppress $|X(f)|^2$ for $f \in \mathcal{F}_{\mathrm{\text{stop}}}$, where $X(f)$ is the discrete-time Fourier transform of $\mathbf{x}$.
Since 
\begin{align*}
    X(f)=\sum_{n=0}^{N-1}x_n e^{-j(2\pi f) \times n}=\mathbf{f}(f)^H \mathbf{x},
\end{align*}
$|X(f)|^2$ can be written as 
\begin{align*}
    |X(f)|^2=\mathbf{x}^H\mathbf{F}(f)\mathbf{x}.
\end{align*}
Therefore, if we want to control the spectrum of $\mathbf{x}$ over $\mathcal{F}_{\mathrm{stop}}$, we can set 
\begin{align}\label{cst:spectral}
\mathbf{x}^H\mathbf{F}(f)\mathbf{x}\leq U_{\mathrm{max}}, \forall f \in\mathcal{F}_{\mathrm{\text{stop}}},
\end{align}
which limits its energy spectral density over $\mathcal{F}_{\mathrm{stop}}$ to $U_{\mathrm{max}}$ maximum.

In practice, instead of considering all points of $\mathcal{F}_{\mathrm{\text{stop}}}$, we uniformly choose sufficiently many $N_f$ points from $\mathcal{F}_{\mathrm{\text{stop}}}$ and take them into our constraints since $\mathcal{F}_{\mathrm{\text{stop}}}$ is uncountable for the most part. 
That is, the constraints \eqref{cst:spectral} correspond to the constraints
\begin{align}\label{cst:spectral_Nf}
\mathbf{x}^H\mathbf{F}(f_i)\mathbf{x}\leq U_{\mathrm{max}}, \forall i \in \mathbb{Z}_{N_f},
\end{align}
where $f_i$'s are $N_f$ points uniformly chosen from $\mathcal{F}_{\mathrm{\text{stop}}}$. 
\begin{rmk}
In fact, $U_{\mathrm{max}}$ does not need to be constant over $\mathcal{F}_{\mathrm{\text{stop}}}$. We can also set different $U_{\mathrm{max}}$'s for different $f_i$'s for the purpose of fitting a spectrum mask.
\end{rmk}

\subsubsection{Unimodular Sequence}\label{sec:unimodular}
Similarly, since it is simpler to deal with a quadratic form than with an absolute value in an optimization problem, we first define the following matrix for the formulation of the quadratic form.
\begin{defn}
Given a positive integer $N$, the $N \times N$ matrix $\mathbf{E}_{n}^{(N)}$ is defined as
$$\mathbf{E}_{n}^{(N)}=\mathbf{e}_n^{(N)}(\mathbf{e}_n^{(N)})^{T},$$
where $\mathbf{e}_n^{(N)}$ is the $n$-th $N$-dimensional standard vector, and $n\in\mathbb{Z}_N$.
\end{defn}
A unimodular sequence $\mathbf{x} \in \mathbb{C}^N$ is of the form $$\mathbf{x}=[e^{j\theta_0}~e^{j\theta_1}~\dots ~ e^{j\theta_{N-1}}]^T,$$ where $\theta_n\in \mathbb{R}$ for all $n \in \mathbb{Z}_{N}.$ Since $n\in \mathbb{Z}_N$, $|x_n|=1$ is tantamount to $|x_n|^2=\mathbf{x}^H\mathbf{E}_{n}^{(N)}\mathbf{x} =1$, a sequence $\mathbf{x}\in \mathbb{C}^N$ is unimodular if and only if
\begin{align} \label{cst:unimodular}
\mathbf{x}^H\mathbf{E}_{n}^{(N)}\mathbf{x} =1, \forall n\in \mathbb{Z}_{N}.
\end{align}

\subsubsection{Problem Reformulation}\label{sec:reformulation}
Combining \eqref{eq:psl}, \eqref{cst:psl_matrix}, \eqref{cst:spectral_Nf}, and \eqref{cst:unimodular}, we have the reformulated problem: 
\begin{subequations}\label{prob:MatrixRepresentationOriginal}
\begin{align}
& \underset{\mathbf{x}\in\mathbb{C}^{N}}{\mathrm{minimize}} && \underset{\ell \in \mathbb{Z}_{N}\backslash \{0\}}{\mathrm{max}}\{\vert\mathbf{x}^H\mathbf{N}_N^{\ell}\mathbf{x}\vert^2\} \label{prob:MatrixRepresentationOriginal_a}\\
& \text{subject to} &&  \mathbf{x}^H\mathbf{F}(f_i)\mathbf{x}\leq U_{\mathrm{max}}, \forall i \in \mathbb{Z}_{N_f}\label{prob:MatrixRepresentationOriginal_b}\\
&&& \mathbf{x}^H\mathbf{E}_{n}^{(N)}\mathbf{x} =1, \forall n\in \mathbb{Z}_{N},\label{prob:MatrixRepresentationOriginal_c}
\end{align}
\end{subequations}
where $f_i$'s are $N_f$ points uniformly chosen from $\mathcal{F}_{\mathrm{\text{stop}}}$. Note that we take the square of the absolute value of the autocorrelation in \eqref{prob:MatrixRepresentationOriginal_a} for the convenience of the further reformulation. This reformulated problem is equivalent to the problem \eqref{prob:original} because the square function is strictly increasing for non-negative inputs.
Then, to eliminate the maximum operator in the objective function, we can reformulate the problem \eqref{prob:MatrixRepresentationOriginal} as its epigraph representation
\begin{subequations}\label{prob:MatrixRepresentation}
\begin{align}
& \underset{\mathbf{x}\in\mathbb{C}^{N}, t\in \mathbb{R}}{\mathrm{minimize}} && t  \label{prob:MatrixRepresentation_a}\\
& \text{subject to} && \mathrm{tr}((\mathbf{N}_{N}^T)^\ell\mathbf{x} \mathbf{x}^H\mathbf{N}_{N}^\ell\mathbf{x}\mathbf{x}^H) \leq t, \forall \ell \in \mathbb{Z}_N\backslash \{0\} \label{prob:MatrixRepresentation_b}\\
&&& \mathbf{x}^H\mathbf{F}(f_i)\mathbf{x}\leq U_{\mathrm{max}}, \forall i \in \mathbb{Z}_{N_f}  \label{prob:MatrixRepresentation_c}\\
&&& \mathbf{x}^H\mathbf{E}_{n}^{(N)}\mathbf{x} =1, \forall n\in \mathbb{Z}_{N},\label{prob:MatrixRepresentation_d}
\end{align}
\end{subequations}
where $f_i$'s are $N_f$ points uniformly chosen from $\mathcal{F}_{\mathrm{\text{stop}}}$. 
\begin{rmk}
From \eqref{cst:psl_matrix}, we can see that $$\mathrm{tr}((\mathbf{N}_{N}^T)^\ell\mathbf{x} \mathbf{x}^H\mathbf{N}_{N}^\ell\mathbf{x}\mathbf{x}^H)=
(\mathbf{x}^H\mathbf{N}_{N}^\ell\mathbf{x})( \mathbf{x}^H{\mathbf{N}_{N}^T}^\ell\mathbf{x})=r_\ell r^*_\ell=|r_\ell|^2.$$
Hence, the constraint \eqref{prob:MatrixRepresentation_b} implies $\operatorname{max}\{|r_\ell|^2\}_{\ell=1}^{N-1} \leq t$, and vice versa.
\end{rmk}

In the following, to address this nonconvex problem \eqref{prob:MatrixRepresentation}, we will introduce a theorem from \cite{Demir2014} and use it to reformulate the problem \eqref{prob:MatrixRepresentation} into an equivalent biconvex problem for usage of the exact penalty approach \cite{Demir2014, Huang2022}. The theorem is introduced as follows with a more straightforward proof than the one given in \cite{Demir2014}.

\begin{thm}\label{thm:TwoMatrixIneq}
For any matrices $\mathbf{A},\mathbf{B} \in \mathbb{H}^n_{+}$, the following inequality always holds.
\begin{equation*}
	\mathrm{tr}(\mathbf{A}\mathbf{B}) \leq \mathrm{tr}(\mathbf{A})\mathrm{tr}(\mathbf{B}).
\end{equation*}
The equality is achieved if and only if $\mathbf{A}$ and $\mathbf{B}$ are linearly dependent and of rank at most one.
\end{thm}
\begin{proof}
The proof of Theorem \ref{thm:TwoMatrixIneq} is given in Appendix \ref{appendix:proof_of_thm_TwoMatrixIneq}.
\end{proof}
Based on Theorem \ref{thm:TwoMatrixIneq}, we can transform the problem \eqref{prob:MatrixRepresentation} into the following equivalent biconvex problem \cite{Demir2014, Huang2022}.
\begin{subequations}\label{prob:equivalent}
\begin{align}
& \underset{{\substack{\mathbf{X}_1\in\mathbb{H}^{N}_{+}, \mathbf{X}_2\in\mathbb{H}^{N}_{+}, t\in \mathbb{R}}}}{\mathrm{minimize}} && t \\
&\quad~\text{subject to} && \mathrm{tr}((\mathbf{N}_{N}^T)^\ell\mathbf{X}_1\mathbf{N}_{N}^\ell\mathbf{X}_2) \leq t, \forall \ell \in \mathbb{Z}_N\backslash \{0\}  \\
&&& \mathrm{tr}(\mathbf{F}(f_i)\mathbf{X}_1)\leq U_{\mathrm{max}}, \forall i\in\mathbb{Z}_{N_f} \\
&&& \mathrm{tr}(\mathbf{F}(f_i)\mathbf{X}_2)\leq U_{\mathrm{max}}, \forall i\in\mathbb{Z}_{N_f}\\
&&& \mathrm{tr}(\mathbf{E}_{n}^{(N)}\mathbf{X}_1) =1, \forall n\in \mathbb{Z}_{N} \label{prob:equivalent_e}\\
&&& \mathrm{tr}(\mathbf{E}_{n}^{(N)}\mathbf{X}_2) =1, \forall n\in \mathbb{Z}_{N} \label{prob:equivalent_f}\\
&&& \mathrm{tr}(\mathbf{X}_1\mathbf{X}_2)=\mathrm{tr}(\mathbf{X}_1)\mathrm{tr}(\mathbf{X}_2), \label{prob:equivalent_g}
\end{align}
where $f_i$'s are $N_f$ points uniformly chosen from $\mathcal{F}_{\mathrm{\text{stop}}}$. 
\end{subequations}
With this equivalent biconvex problem \eqref{prob:equivalent}, the exact penalty approach can be utilized to give an approximated solution as shown in the next section. 
\begin{rmk}\label{rmk:equivalence}
To see the equivalence between the problem \eqref{prob:MatrixRepresentation} and the problem \eqref{prob:equivalent}, first notice that the constraints \eqref{prob:equivalent_e}, \eqref{prob:equivalent_f} and \eqref{prob:equivalent_g} imply $\mathbf{X}_1=\mathbf{X}_2$ and $\mathrm{rank}(\mathbf{X}_1) = \mathrm{rank}(\mathbf{X}_2) = 1$ according to Theorem \ref{thm:TwoMatrixIneq}. Hence, since $\mathbf{X}_1$ and $\mathbf{X}_2$ are both positive semidefinite, we can decompose them as $\mathbf{X}_1=\mathbf{X}_2=\mathbf{x}\mathbf{x}^H$ for some $\mathbf{x} \in \mathbb{C}^N$. Then, the constraints in the problem \eqref{prob:equivalent} can be reduced to those in the problem \eqref{prob:MatrixRepresentation}.
\end{rmk}

\subsection{Exact Penalty Approach (SDR)}\label{sec:optimization_algorithm}
The non-convex constraints in the problem \eqref{prob:equivalent} hinder us from developing an efficient algorithm for the optimization problem. However, the problem \eqref{prob:equivalent} is a biconvex optimization problem which can be addressed by alternating minimization \cite{Demir2014}. To relax the constraint \eqref{prob:equivalent_g}, which makes the problem have no space for alternating minimization, we first introduce the nonnegative penalty function \cite{Demir2014, Huang2022}
\begin{align}
    \mathrm{tr}(\mathbf{X}_1)\mathrm{tr}(\mathbf{X}_2)-\mathrm{tr}(\mathbf{X}_1\mathbf{X}_2)\label{original_penalty_function}
\end{align}
whose value is zero if and only if the constraint \eqref{prob:equivalent_g} holds. Note that the constraints \eqref{prob:equivalent_e} and \eqref{prob:equivalent_f} imply that $\mathrm{tr}(\mathbf{X}_1) = \mathrm{tr}(\mathbf{X}_2) = N$, and thus \eqref{original_penalty_function} can be simplified as
\begin{align}
    N^2-\mathrm{tr}(\mathbf{X}_1\mathbf{X}_2).\label{simplified_penalty_function}
\end{align}
Then, with the help of the penalty function, we can apply alternating minimization to break the problem \eqref{prob:equivalent} into two subproblems as follows.
\begin{subequations}\label{prob:AM_first}
\begin{align}
& \underset{\mathbf{X}_1\in\mathbb{H}^{N}_{+}, t\in \mathbb{R}}{\mathrm{minimize}} && (1-w)t + w\big[N^2-\mathrm{tr}(\mathbf{X}_1\mathbf{X}_2)\big]^2\label{prob:AM_first_a}\\
&\text{subject to} && \mathrm{tr}((\mathbf{N}_{N}^T)^\ell\mathbf{X}_1\mathbf{N}_{N}^\ell\mathbf{X}_2) \leq t, \forall \ell \in \mathbb{Z}_N\backslash \{0\}  \\
&&& \mathrm{tr}(\mathbf{F}(f_i)\mathbf{X}_1)\leq U_{\mathrm{max}}, \forall i \in\mathbb{Z}_{N_f}\\
&&& \mathrm{tr}(\mathbf{E}_{n}^{(N)}\mathbf{X}_1) =1, \forall n\in \mathbb{Z}_{N},
\end{align} 
\end{subequations}
where $f_i$'s are $N_f$ points uniformly chosen from $\mathcal{F}_{\mathrm{\text{stop}}}$, and $\mathbf{X}_2$ is a constant matrix. 
\begin{subequations}\label{prob:AM_second}
\begin{align}
& \underset{\mathbf{X}_2\in\mathbb{H}^{N}_{+}, t\in \mathbb{R}}{\mathrm{minimize}} && (1-w)t - w\cdot\mathrm{tr}(\mathbf{X}_1\mathbf{X}_2)\hspace{2.4cm}\label{prob:AM_second_a}\\
&\text{subject to} && \mathrm{tr}((\mathbf{N}_{N}^T)^\ell\mathbf{X}_1\mathbf{N}_{N}^\ell\mathbf{X}_2) \leq t, \forall \ell \in \mathbb{Z}_N\backslash \{0\}  \\
&&& \mathrm{tr}(\mathbf{F}(f_i)\mathbf{X}_2)\leq U_{\mathrm{max}}, \forall i\in\mathbb{Z}_{N_f} \\
&&& \mathrm{tr}(\mathbf{E}_{n}^{(N)}\mathbf{X}_2) =1, \forall n\in \mathbb{Z}_{N},
\end{align}
\end{subequations}
where $f_i$'s are $N_f$ points uniformly chosen from $\mathcal{F}_{\mathrm{\text{stop}}}$, and $\mathbf{X}_1$ is a constant matrix.
In the two subproblems \eqref{prob:AM_first}, \eqref{prob:AM_second}, $w \in [0,1]$ is a constant that controls the relative weights on the penalty functions, which makes the two subproblems tend to meet the constraint \eqref{prob:equivalent_g}. Note that the penalty functions in both subproblems are slightly different from \eqref{simplified_penalty_function}. We will give the reasons later in this section. 

First, because $N^2-\mathrm{tr}(\mathbf{X}_1\mathbf{X}_2)$ is affine in $\mathbf{X}_1$, and the function $f(x)=x^2$ is convex , we know that the function $\big[N^2-\mathrm{tr}(\mathbf{X}_1\mathbf{X}_2)\big]^2$ is convex in $\mathbf{X}_1$. Therefore, the problem \eqref{prob:AM_first} is a convex optimization problem.
Since both the problem \eqref{prob:AM_first} and the problem \eqref{prob:AM_second} are convex optimization problems, they can be solved efficiently via CVX, a package for specifying and solving convex programs \cite{cvx}. The main idea of our algorithm is to alternately solve the problem \eqref{prob:AM_first} and the problem \eqref{prob:AM_second} in pursuit of an approximately optimal solution to the non-convex problem \eqref{prob:equivalent}. Our algorithm is given in Algorithm \ref{algo:AM}. Notice that, in Algorithm \ref{algo:AM},  we alternatively solve the two subproblems, whose differences are shown in their penalty functions. The penalty function in \eqref{prob:AM_first_a} is the square of \eqref{simplified_penalty_function}, which can make the algorithm converge faster to a rank-one solution than simply applying \eqref{simplified_penalty_function} as a penalty function. The penalty function in \eqref{prob:AM_second_a} is directly from \eqref{simplified_penalty_function} with constant $N^2$ omitted, which makes Algorithm \ref{algo:AM} focus more on finding a good solution with small PSL in step \ref{step:solve_AM_second} than in step \ref{step:solve_AM_first}. With this algorithm, which alternately focuses on making the solution into a rank-one solution and finding a good solution with small PSL, we found it converge within an acceptable time to a satisfactory solution\footnote{The parameters $\varepsilon_{\mathrm{rank}}$ should be small enough to prevent the projection in step \ref{step:proj_to_vec} from resulting in an erroneous solution which may not satisfy the constraints in the problem \eqref{prob:AM_second}.}, which is demonstrated in Section \ref{sec:NumericalValidation}  and evaluated in Section \ref{sec:PerformanceAnalysis}.

\algnewcommand{\Initialization}[1]{%
  \State \textbf{Initialization}
}
\algnewcommand{\Repeatu}[1]{%
  \State \textbf{Repeat}
}
\algnewcommand{\until}[1]{%
  \State \textbf{until}
}
\begin{algorithm}
\caption{AM Method for PSL Minimizing Problem with Maximal Spectral and Unimodular Constraints}
\algrenewcommand\algorithmicprocedure{\textbf{Initialization}}
\textbf{Input} $N, N_f, w, U_{\mathrm{max}}, \varepsilon_{x}, \varepsilon_{\mathrm{rank}}, \varphi_{\mathrm{max}}, \mathcal{F}_{\mathrm{\text{stop}}}, \mathbf{x}_{\mathrm{init}}$\\
\textbf{Output} The vector $\mathbf{x}$ of the optimal pulse
\begin{algorithmic}[1]
\Initialization
\State \quad Let $\mathbf{X}_2^{(0)}=\mathbf{x}_{\mathrm{init}}\mathbf{x}_{\mathrm{init}}^H,~\varphi = 0$.
\State \quad Uniformly choose $N_f$ points $f_i$'s from $\mathcal{F}_{\text{stop}}$ and calculate the corresponding $\{\mathbf{F}(f_i)\}_{i\in\mathbb{Z}_{N_f}}$.
\State \quad Compute the corresponding $\mathbf{N}_{N}, \{\mathbf{E}^{(N)}_n\}_{n\in\mathbb{Z}_N}$ with given $N$.
\Repeatu

\State $\quad$ Solve the problem \eqref{prob:AM_first} for $\mathbf{X}_1^{(\varphi+1)}$ by CVX while fixing $\mathbf{X}_2$ as $\mathbf{X}_2^{(\varphi)}$.\label{step:solve_AM_first}
\State $\quad$ Solve the problem \eqref{prob:AM_second} for $\mathbf{X}_2^{(\varphi+1)}$ by CVX while fixing $\mathbf{X}_1$ as $\mathbf{X}_1^{(\varphi+1)}$.\label{step:solve_AM_second}
\State  $\quad$$~\varphi\leftarrow \varphi + 1$

\until \State $\frac{\mathrm{tr}(\mathbf{X}_1)\mathrm{tr}(\mathbf{X}_2)-\mathrm{tr}(\mathbf{X}_1\mathbf{X}_2)}{\mathrm{tr}(\mathbf{X}_1){\mathrm{tr}}(\mathbf{X}_2)}=1-\frac{\mathrm{tr}(\mathbf{X}_1\mathbf{X}_2)}{N^2}\leq \varepsilon_x$ or $\varphi\geq\varphi_{\mathrm{max}}.$

\State Perform the singular value decomposition (SVD) on $\mathbf{X}_2^{(\varphi )}$: $\mathbf{X}_2^{(\varphi )} = \mathbf{U}\Sigma \mathbf{V}^H$
    
\State Set $\sigma_0=
    \begin{bmatrix}
    \Sigma
    \end{bmatrix}_{0,0}$ and $\sigma_1=
    \begin{bmatrix}
    \Sigma
    \end{bmatrix}_{1,1}$.
    
\If {$\sigma_1/\sigma_0\leq \varepsilon_{\mathrm{rank}}$}
    \State Obtain $\mathbf{x}=\sqrt{\sigma_0}\cdot\mathbf{U}\mathbf{e}_0^{N}$.\label{step:proj_to_vec}
\Else
\State Declare failure of convergence.
    
\EndIf
\end{algorithmic}\label{algo:AM}
\end{algorithm}


\section{Numerical Validation}\label{sec:NumericalValidation}
\subsection{Definitions of Parameters}
Since the PSL defined in \eqref{eq:psl} varies along with the sequence length $N$, it is desirable to define the normalized PSL by the ratio of the PSL and the total power of the sequence \cite{he_li_stoica_2012, Cui2018, He2010, Mohammad2021}, which is shown as follows. 

\begin{defn}\label{def:NormalizedPSL}
Given a sequence $\mathbf{x} \in \mathbb{C}^N$, the normalized peak sidelobe level (NPSL) of $\mathbf{x}$ (in decibel) is defined as 
\begin{equation*}
\begin{aligned}
\underset{\ell\in \mathbb{Z}_N\backslash\{0\}}{\mathrm{max}} 20~\mathrm{log}_{10}\left(\left\vert\frac{r_\ell}{r_0}\right\vert\right)=\underset{\ell\in \mathbb{Z}_N\backslash\{0\}}{\mathrm{max}} 20~\mathrm{log}_{10}\left(\frac{\vert\mathbf{x}^H(\mathbf{N}_{N}^\ell)^T\mathbf{x}\vert}{N}\right).
\end{aligned}\end{equation*}
\end{defn}

To demonstrate the quality of any sequence in spectral compatibility, we need to define the average passband energy and the maximal stopband energy as follows.
\begin{defn}\label{def:APB}
The average passband energy is defined as
\begin{align*}
E_{\text{APB}} = \frac{1}{B_{\mathrm{pass}}}\int_{\mathcal{F}_{\mathrm{pass}}}\vert X(f) \vert^2 \, df,
\end{align*}
where $\mathcal{F}_{\mathrm{pass}}=[0,1]\backslash \mathcal{F}_{\mathrm{stop}}\subseteq [0,1]$ is the passband region and $B_{pass}=\int_{F_{pass}}df$ is the passband bandwidth.
\end{defn}
\begin{defn}\label{def:MSB}
The maximal stopband energy is defined as
\begin{align*}
E_{\text{MSB}} = \underset{{f\in \mathcal{F}_{\mathrm{stop}}}}{\mathrm{max}}\vert X(f) \vert^2.
\end{align*}
\end{defn}
With the $E_{\text{APB}}$ and $E_{\text{MSB}}$ defined above, we can define the stopband attenuation as follows \cite{he_li_stoica_2012,He2010}.
\begin{defn}\label{def:StopbandAttenuation}
The stopband attenuation $A_{\text{stop}}$ (in decibel) is defined as
\begin{equation*}
\begin{aligned}
A_{\text{stop}} = 10\cdot\mathrm{log_{10}}\left(\frac{E_{\text{APB}}}{E_{\text{MSB}}}\right),
\end{aligned}
\end{equation*}
where $E_{\text{APB}}$ is the average passband energy and $E_{\text{MSB}}$ is the maximal stopband energy.
\end{defn}

Note that, practically, we can only approximate the value of $A_{\mathrm{stop}}$ by uniformly calculating a finite number of samples in the frequency response $X(f)$ to approximate $E_{\mathrm{APB}}$ and $E_{\mathrm{MSB}}$.
Specifically, we can calculate the mean of the energy spectral density of the samples within $\mathcal{F}_{\mathrm{pass}}$ to approximate $E_{\mathrm{APB}}$. For $E_{\mathrm{APB}}$, it can be approximated by finding the maximal energy spectral density of the samples within $\mathcal{F}_{\mathrm{stop}}$.

\subsection{Choice of $\boldsymbol{U_{\mathrm{max}}}$}
With appropriate $U_{\mathrm{max}}$, we can achieve the desired stopband attenuation. 
Their relations are shown as follows:
\subsubsection{Choice such that $A_{\mathrm{stop}}\geq A$}
We can set $$U_{\mathrm{max}}=\frac{N}{10^{0.1\times A}}$$ to achieve $A_{\mathrm{stop}} \geq A.$ The reason arises from the fact that 
\begin{align}
    \frac{E_{\text{APB}}}{E_{\text{MSB}}} > \frac{N}{E_{\text{MSB}}} \geq \frac{N}{U_{\mathrm{max}}}. \label{Umax1_ineq}
\end{align} Thus, we have $A_{\mathrm{stop}}\geq  10 ~\mathrm{log_{10}}\frac{N}{U_{\mathrm{max}}}=A$. The first inequality in \eqref{Umax1_ineq} is due to the fact that the average passband energy $E_{\text{APB}}$ is larger than the overall average energy, which is 
$\int_{0}^1\vert X(f)\vert ^2 df = \sum_{n=0}^{N-1}\vert x_n \vert^2 = N$ because of the Parsevel's relation. 

\subsubsection{Choice such that $A_{\mathrm{stop}}\approx A$} 
We can set $$U_{\mathrm{max}}=\frac{N}{10^{0.1\times A}}\times \frac{1}{B_{\mathrm{pass}}}$$ to achieve $A_{\mathrm{stop}}\approx A.$ The reason comes from the fact that $E_{\text{APB}} =\frac{N -\epsilon}{B_{\text{pass}}}$, where $N$ is the overall energy, and $\epsilon$ is the total energy of the stopband. Therefore, we have $A_{\text{stop}}= 10~\mathrm{log_{10}}\frac{E_{\text{APB}}}{E_{\text{MSB}}}
 \approx  10~\mathrm{log_{10}}\frac{ N -\epsilon}{(U_{\text{max}} \times B_{\text{pass}})} 
 \approx A.$
Here, we assume that the maximal stopband energy $E_{\text{MSB}}$ is approximately equal to $U_{\text{max}}$, and $\epsilon$ is sufficiently small. The accuracy of this approximation depends on the total energy of stopband.
\subsection{Numerical Examples}\label{Sec:Numerical_Examples}
Numerical examples were carried out to illustrate the advantages of the proposed Algorithm \ref{algo:AM} (labeled as ``Proposed Algorithm'') over the SCAN algorithm proposed in \cite{He2010} (labeled as ``SCAN'') and the MM-PMM algorithm proposed in \cite{LU2022}  (labeled as ``MM-PMM'') when it comes to the normalized peak side-lobe level and stopband attenuation. 
Unless otherwise specified, we use the SCAN-generated sequence as the initial point $\mathbf{x}_{\mathrm{init}}$ and take $500$ samples to plot the normalized energy spectrum in such a way that $E_{\mathrm{APB}}$ is equal to 1 \cite{he_li_stoica_2012,He2010}. 
All numerical examples in this paper were conducted by MATLAB 2020a on a personal computer with AMD Ryzen 9 3900X 12-core processor and 32 GB RAM.
The five cases are presented below.
\subsubsection{Case 1}\label{sec:example_case1}
First, we designed a unimodular sequence with $N=32$ and $\mathcal{F}_{\text{stop}}=[0.2,0.3]$.
We run the SCAN algorithm with the parameters $\tilde{N}=10N,$ $\lambda=0.97$ and the MM-PMM algorithm with the parameters $\gamma_r=20,\gamma_y=20,\rho_r=0.1,\rho_y=5000,L_{\mathrm{max}}=430,\eta_r=0.0001,\eta_y=0.0001,p=20$ to generate two sequences. 
Next, we applied Algorithm $\ref{algo:AM}$ to generate the sequence with the parameters $w=0.11,$ $U_{\mathrm{max}}=0.032,$ $\varepsilon_{x}=2\times 10^{-3},$ $\varepsilon_{\mathrm{rank}}=10^{-8},$  $\varphi_{\mathrm{max}}=50,$ $\mathcal{F}_{\text{stop}}=[0.2,0.3],$ $N_{f} = 30.$ 
Here, $f_i$ is chosen as $f_i=0.2+\frac{0.1}{30}\times i$ for $i\in \mathbb{Z}_{30}.$ 
The normalized auto-correlations and energy spectra of these sequences are shown in Fig. \ref{figN32F2030}, and their normalized peak sidelobe levels and stopband attenuation are summarized in Table \ref{tableN32F2030}. 
We can see from Table \ref{tableN32F2030} that our proposed method achieves a better NPSL ($-18.18$ dB) than that of the MM-PMM algorithm ($-14.67$ dB) and that of the SCAN algorithm ($-16.52$ dB) while enjoying controllable stopband attenuation.
\begin{table}[hbtp!]
  \centering
  \caption{NPSL and $A_{\mathrm{stop}}$ for different methods.}
    \label{tableN32F2030}
  \begin{threeparttable}
    \begin{tabular}{c c c}
      \toprule
      \textbf{Method}           & \textbf{NPSL} (dB) & \textbf{Stopband  Attenuation} (dB) \\
      \midrule
      SCAN          & -16.52  & 17.92 \\
      MM-PMM &-14.67  &15.41\\
      Proposed Algorithm &   -18.18   & 30.24\\
      \bottomrule
    \end{tabular}
  \end{threeparttable}
    
\end{table}
\begin{figure}[hbtp!]
    \centering
    \includegraphics[width=0.48\textwidth]{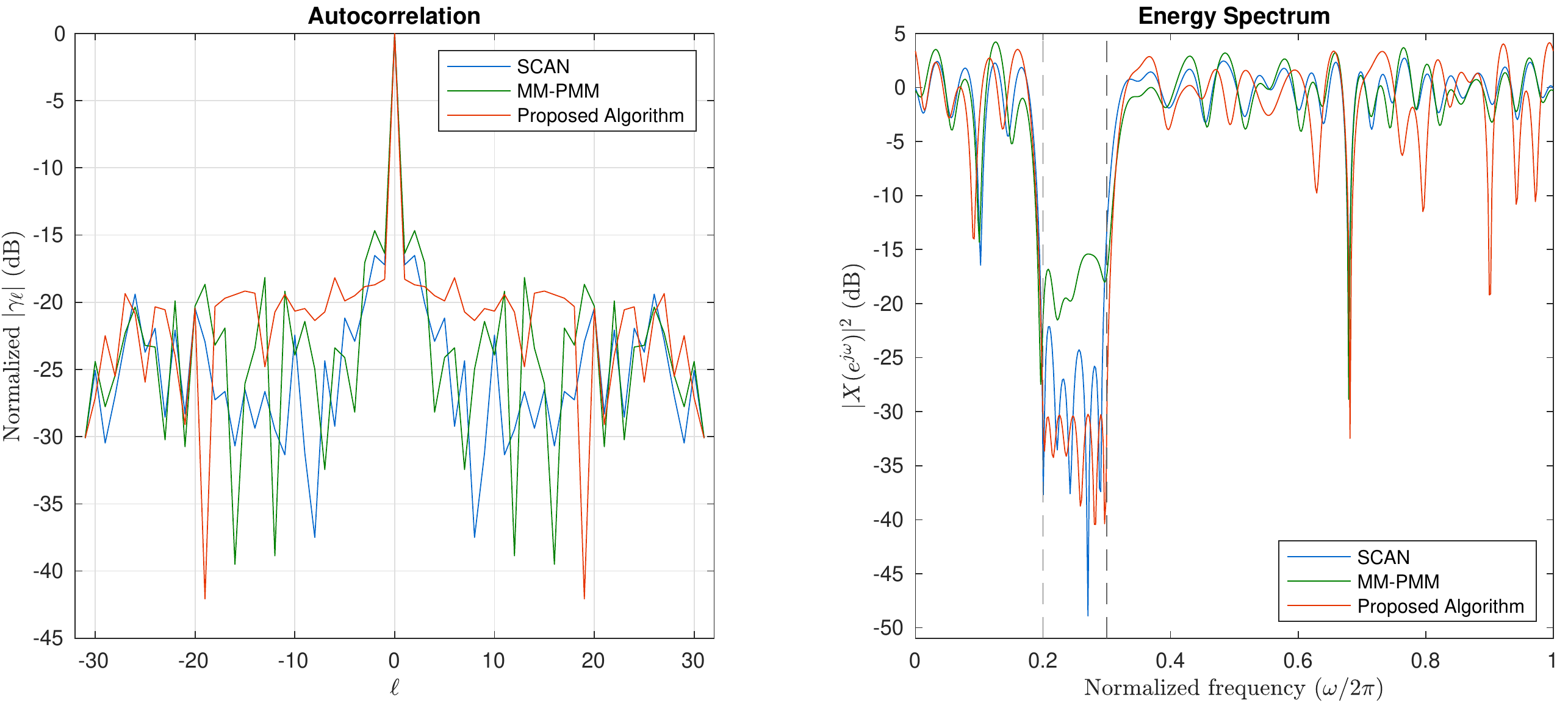}
    \caption{The auto-correlations and energy spectra comparison of sequences in Case 1 ($N=32$)}
    \label{figN32F2030}
\end{figure}

\subsubsection{Case 2}\label{sec:case2}
Second, we designed a unimodular sequence with larger $N=100$ and $\mathcal{F}_{\text{stop}}=[0.2,0.3]$. 
The SCAN algorithm was used with the parameters  $\tilde{N}=10N,$ $\lambda=0.97$ to generate a sequence. 
Next, the MM-PMM algorithm was run with the parameters $\gamma_r=10,\gamma_y=10,\rho_r=0.1,\rho_y=50000,L_{\mathrm{max}}=265,\eta_r=0.0001,\eta_y=0.0001,p=20$ to generate another sequence.
Then, Algorithm \ref{algo:AM} was applied with the parameters $w=0.11,$ $U_{\mathrm{max}}=0.1,$ $\varepsilon_{x}=2\times 10^{-3},$ $\varepsilon_{\mathrm{rank}}=10^{-8},$ $\varphi_{\mathrm{max}}=50,$ $\mathcal{F}_{\text{stop}}=[0.2,0.3],$ $N_{f} = 60$ to obtain the sequence. 
Here, $f_i$ is chosen as $f_i=0.2+\frac{0.1}{60}\times i$ for $i\in \mathbb{Z}_{60}.$
The normalized auto-correlations and energy spectra of these sequences are plotted in Fig. \ref{figN100F2030}, and their normalized peak side-lobe levels and stopband attenuation are summarized in Table \ref{tableN100F2030}. It can be seen from Table \ref{tableN100F2030} that our method has the advantage of NPSL ($-21.16$ dB) over the MM-PMM algorithm ($-16.20$ dB) and the SCAN algorithm ($-17.91$ dB) while restricting $X(f)$ to lower than $-30$ dB over $\mathcal{F}_{\text{stop}}$.
\begin{table}[hbtp!]
  \centering
  \begin{threeparttable}
    \caption{NPSL and $A_{\mathrm{stop}}$ for different methods.}
    \label{tableN100F2030}
    \begin{tabular}{c c c}
      \toprule
      \textbf{Method}           & \textbf{NPSL} (dB) & \textbf{Stopband  Attenuation} (dB) \\
      \midrule
      SCAN          & -17.92  & 22.50 \\
      MM-PMM & -16.20 &7.07\\
      Proposed Algorithm &  -21.16   & 30.14\\
      \bottomrule
    \end{tabular}
  \end{threeparttable}
\end{table}
\begin{figure}[hbtp!]
    \centering
    \includegraphics[width=0.48\textwidth]{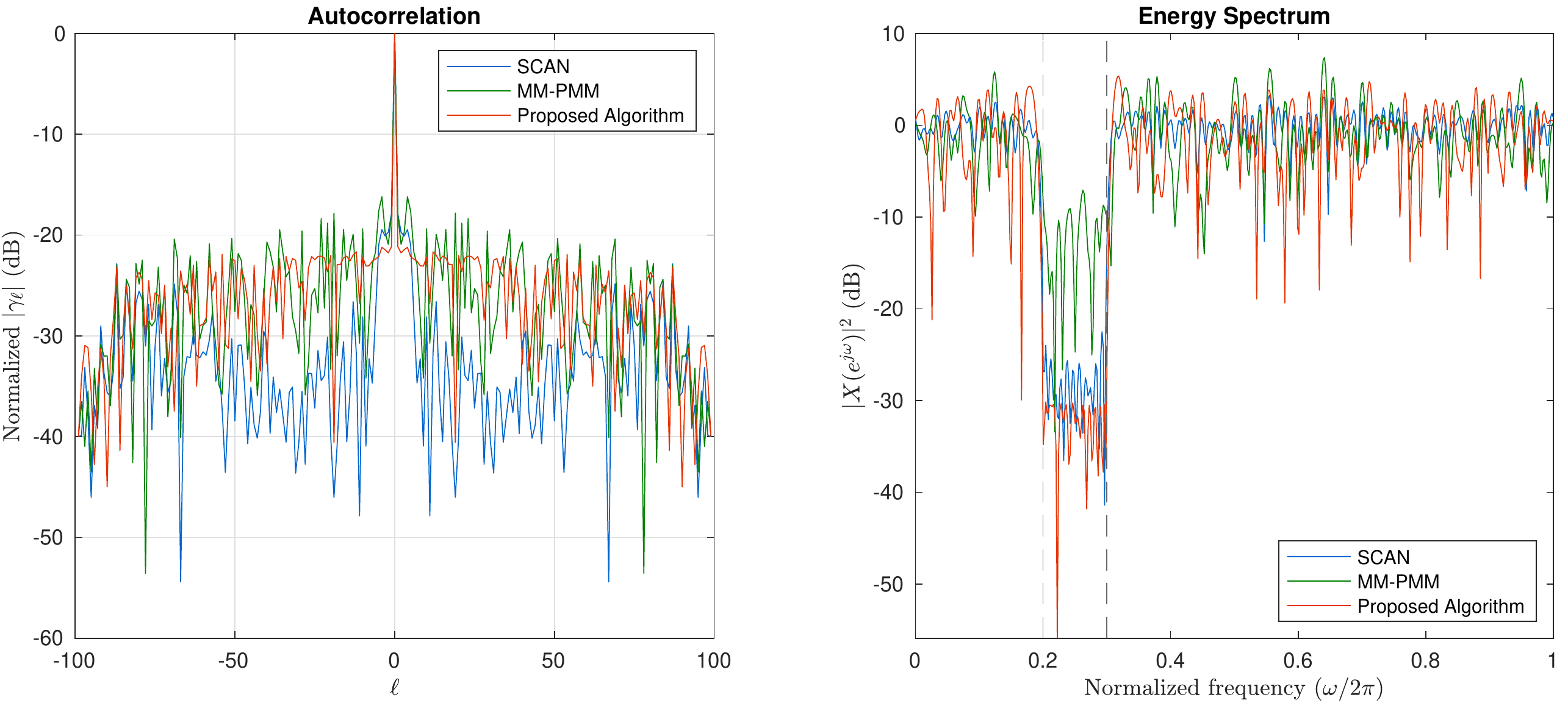}
    \caption{The auto-correlations and energy spectra comparison of sequences in Case 2 ($N=100$)}
    \label{figN100F2030}
\end{figure}

\subsubsection{Case 3}
Third, we designed a sequence for the notch with $N=100$ and $\mathcal{F}_{\text{stop}}=[0.6,0.62].$ 
The SCAN algorithm was applied with the parameters $\tilde{N}=10N,$ $\lambda=0.97$ to generate a sequence. 
Next, the MM-PMM algorithm was run with the parameters $\gamma_r=10,\gamma_y=10,\rho_r=0.1,\rho_y=50000,L_{\mathrm{max}}=265,\eta_r=0.0001,\eta_y=0.0001,p=20$ to generate another sequence.
Then, we run Algorithm \ref{algo:AM} to obtain the sequence with the parameters $w=0.18,$ $U_{\mathrm{max}}=0.001,$ $\varepsilon_{x}=10^{-4},$ $\varepsilon_{\mathrm{rank}}=10^{-8},$ $\varphi_{\mathrm{max}}=50,$ $\mathcal{F}_{\text{stop}}=[0.6,0.62],$ $N_{f} = 50$. 
Here, $f_i$ is chosen as $f_i=0.6+\frac{0.02}{50}\times i$ for $i\in \mathbb{Z}_{50}.$ 
The normalized auto-correlations and energy spectra of these  sequences are illustrated in Fig. \ref{figN100F6062}, and their normalized peak sidelobe levels and stopband attenuation are summarized in Table \ref{tableN100F6062}. We can see from Table \ref{tableN100F6062} that our proposed method exceeds the MM-PMM algorithm and the SCAN algorithm in both NPSL and stopband attenuation.

\begin{table}[hbtp!]
  \centering
  \begin{threeparttable}
    \caption{NPSL and $A_{\mathrm{stop}}$ for different methods.}
    \label{tableN100F6062}
    \begin{tabular}{c c c}
      \toprule
      
      \textbf{Method}           & \textbf{NPSL} (dB) & \textbf{Stopband  Attenuation} (dB) \\
      \midrule
      SCAN          & -25.08  & 25.43 \\
      MM-PMM  &-20.86 & 5.95\\
      Proposed Algorithm &  -26.88   & 50.05\\
      \bottomrule
    \end{tabular}
  \end{threeparttable}
\end{table}
\begin{figure}[hbtp!]
    \centering
    \includegraphics[width=0.48
    \textwidth]{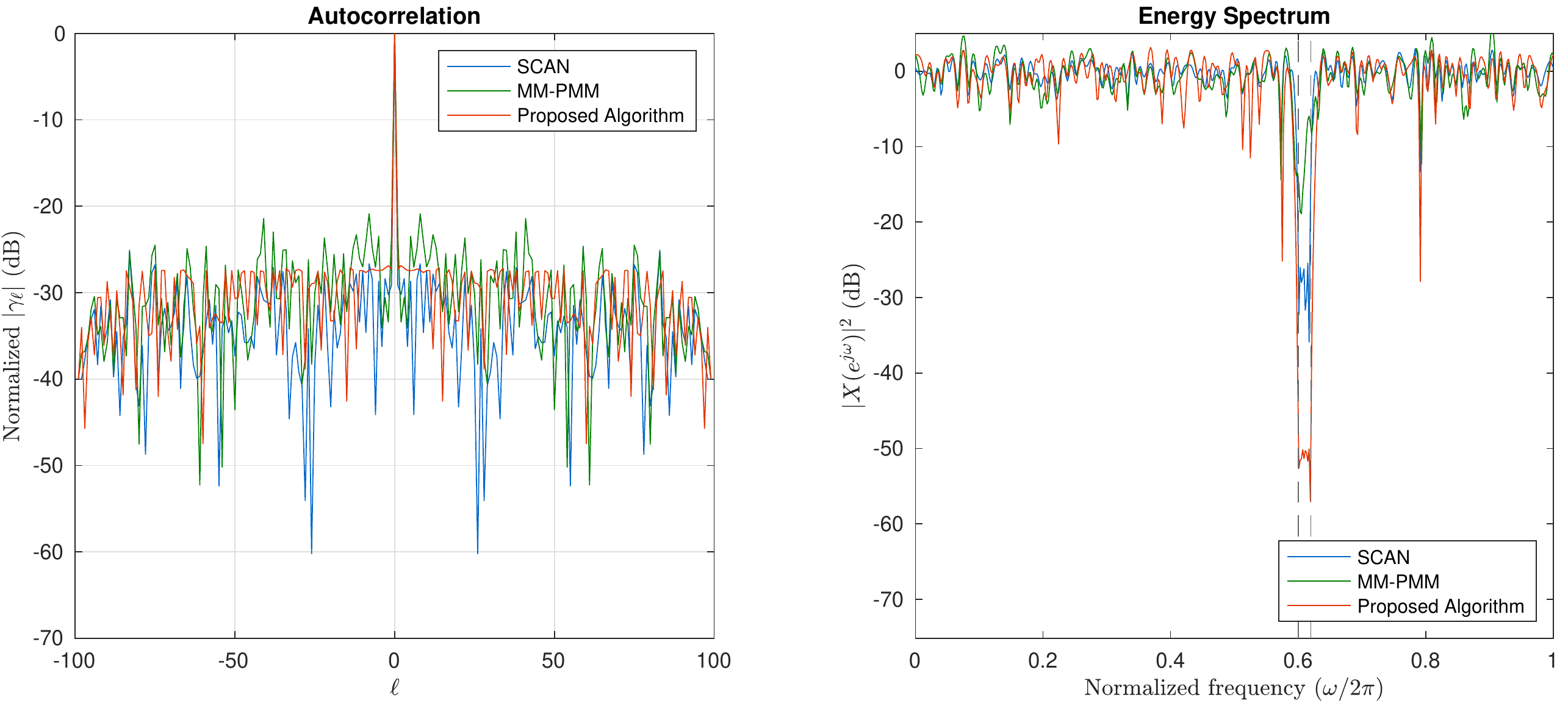}
    \caption{The auto-correlations and energy spectra comparison of sequences in Case 3 ($N=100$)}
    \label{figN100F6062}
\end{figure}

\subsubsection{Case 4}
We also designed a sequence for the notch with $N=128$ and $\mathcal{F}_{\text{stop}}=[0.6,0.62].$ 
The SCAN algorithm with the parameters $\tilde{N}=10N,$ $\lambda=0.97$ was run to generate a sequence. 
Next, the MM-PMM algorithm was applied with the parameters $\gamma_r=10,\gamma_y=10,\rho_r=0.1,\rho_y=50000,L_{\mathrm{max}}=255,\eta_r=0.0001,\eta_y=0.0001,p=20$ to generate another sequence.
Then, we applied Algorithm \ref{algo:AM} with the parameters $w=0.11,$ $U_{\mathrm{max}}=0.00128,$ $\varepsilon_{x}=10^{-4},$ $\varepsilon_{\mathrm{rank}}=10^{-8},$ $\varphi_{\mathrm{max}}=50,$ $\mathcal{F}_{\text{stop}}=[0.6,0.62],$ $N_{f} = 64$ to obtain the sequence. 
Here, $f_i$ is chosen as $f_i=0.6+\frac{0.02}{64}\times i$ for $i\in \mathbb{Z}_{64}.$ 
The normalized auto-correlations and energy spectra of these sequences are plotted in Fig. \ref{figN128F6062}, and their normalized peak sidelobe levels and stopband attenuation are summarized in Table \ref{tableN128F6062}. 
We can see from Table \ref{tableN128F6062} that our proposed method outperforms the MM-PMM algorithm and the SCAN algorithm in terms of NPSL and stopband attenuation.
\begin{table}[hbtp!]
  \centering
  \begin{threeparttable}
    \caption{NPSL and $A_{\mathrm{stop}}$ for different methods.}
    \label{tableN128F6062}
    \begin{tabular}{c c c}
      \toprule
      \textbf{Method}           & \textbf{NPSL} (dB) & \textbf{{}Stopband  Attenuation} (dB) \\
      \midrule
      SCAN          & -25.56   & 20.84 \\
      MM-PMM  & -24.58 &11.28\\
      Proposed Algorithm &  -28.69   & 50.70 \\
      \bottomrule
    \end{tabular}
  \end{threeparttable}
\end{table}
\begin{figure}[hbtp!]
    \centering
    \includegraphics[width=0.48\textwidth]{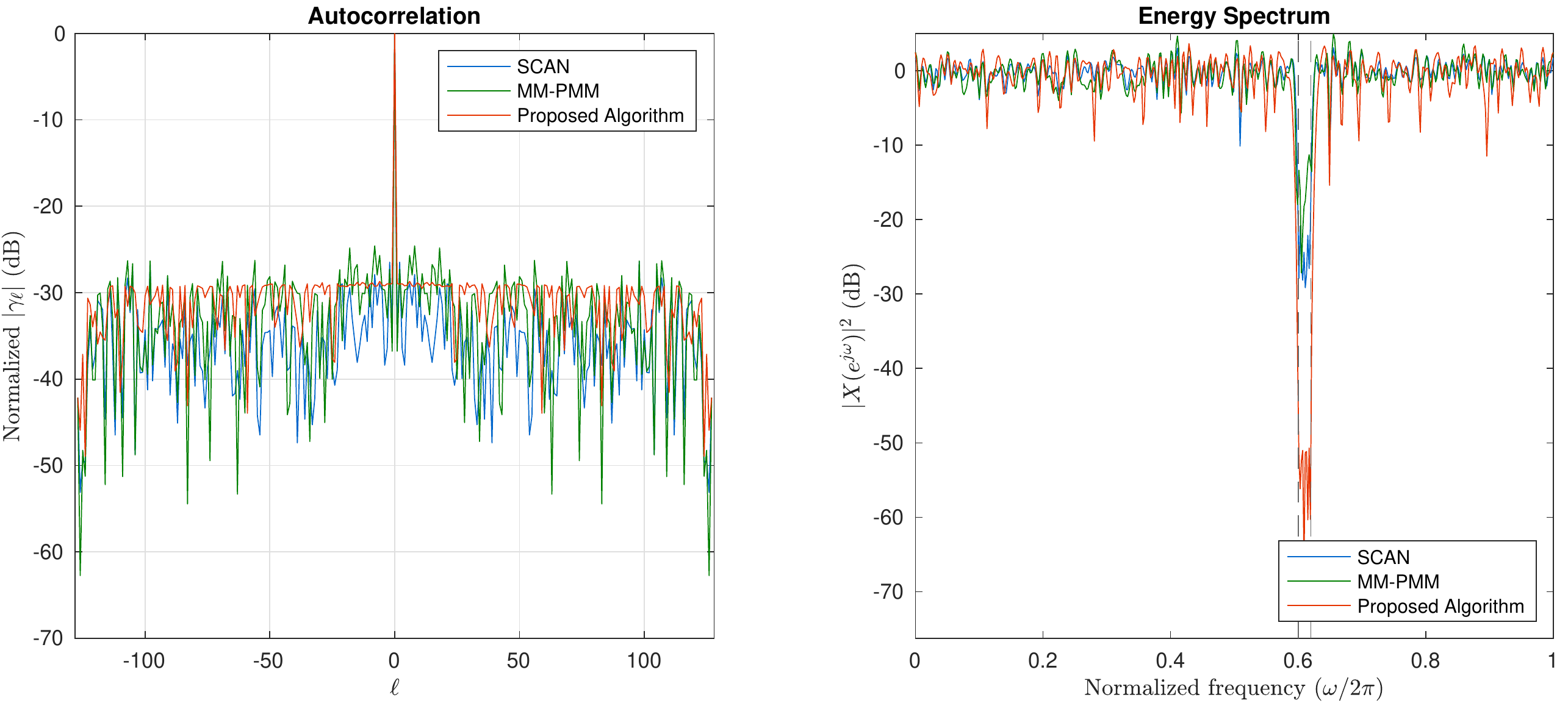}
    \caption{The auto-correlations and energy spectra comparison of sequences in Case 4 ($N=128$)}
    \label{figN128F6062}
\end{figure}
\subsubsection{Case 5}
Finally, we designed a long sequence with $N=256$ and $\mathcal{F}_{\text{stop}}=[0.2,0.3]$ to illustrate the applicability of our method for a long sequence. 
The SCAN algorithm with the parameters $\tilde{N}=10N,$ $\lambda=0.97$ was used to generate the initial vector. 
Next, the MM-PMM algorithm was run with the parameters $\gamma_r=10,\gamma_y=10,\rho_r=0.1,\rho_y=50000,L_{\mathrm{max}}=265,\eta_r=0.0001,\eta_y=0.0001,p=20$ to generate another sequence.
Then, we applied Algorithm \ref{algo:AM} with the parameters $w=0.11,$ $U_{\mathrm{max}}=0.256,$ $\varepsilon_{x}=10^{-3},$ $\varepsilon_{\mathrm{rank}}=10^{-8},$ $\varphi_{\mathrm{max}}=50,$ $\mathcal{F}_{\text{stop}}=[0.2,0.3],$ $N_{f} =150 $ to obtain the desired sequence. 
Here, $f_i$ is chosen as $f_i=0.2+\frac{0.1}{150}\times i$ for $i\in \mathbb{Z}_{150}.$ 
The normalized auto-correlations and energy spectra of these sequences are plotted in Fig. \ref{figN256F2030}, and their normalized peak sidelobe levels and stopband attenuation are summarized in Table \ref{tableN256F2030}. 
We can see from Table \ref{tableN256F2030} that our proposed method outperforms the MM-PMM algorithm and the SCAN algorithm in terms of NPSLs and stopband attenuation.

\begin{table}[hbtp!]
  \centering
  \begin{threeparttable}
    \caption{NPSL and $A_{\mathrm{stop}}$ for different methods.}
    \label{tableN256F2030}
    \begin{tabular}{c c c}
      \toprule
      \textbf{Method}           & \textbf{NPSL} (dB) & \textbf{{}Stopband  Attenuation} (dB) \\
      \midrule
      SCAN          & -18.72   & 26.52 \\
      MM-PMM  & -18.49& 15.93 \\
      Proposed Algorithm &  -22.40   & 30.30 \\
      \bottomrule
    \end{tabular}
  \end{threeparttable}
\end{table}
\begin{figure}[hbtp!]
    \centering
    \includegraphics[width=0.48\textwidth]{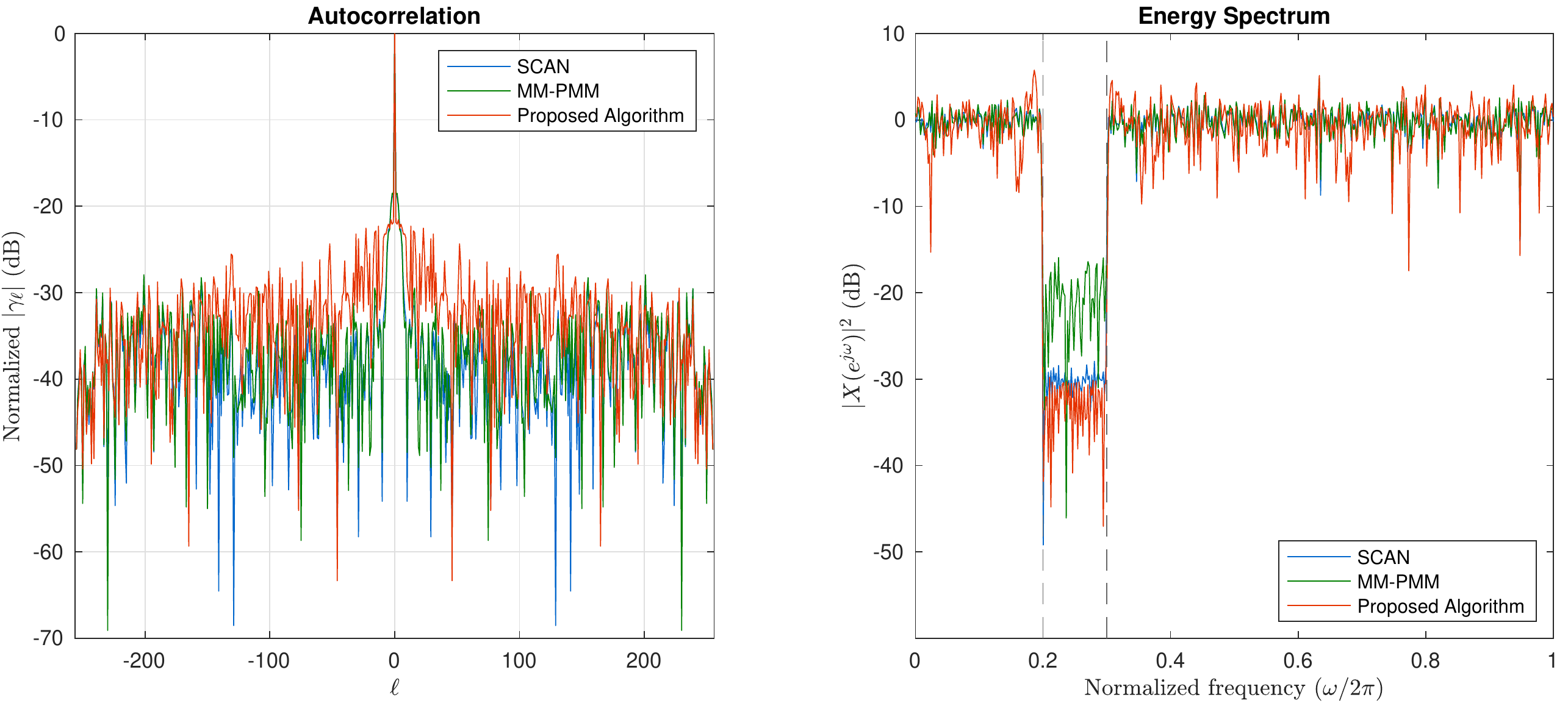}
    \caption{The auto-correlations and energy spectra comparison of sequences in Case 5 ($N=256$)}
    \label{figN256F2030}
\end{figure}
From these numerical results, we can find that the proposed algorithm is superior to the SCAN and the MM-PMM algorithm in terms of NPSL and stopband attenuation.
As for spectral compatibility, this satisfactory result firstly arises from our spectral constraints designed to limit the maximal stopband energy, while the SCAN algorithm only controls the integrated stopband energy. Secondly, the resolution of our spectral constraints is adjustable and can be large even for short sequences. In contrast, the spectral resolution of the other two algorithms is limited by the sequence length and the stopband's width. Thus, their stopband energy tends to fluctuate, thereby having higher stopband attenuation. 
As for auto-correlation, since the SCAN algorithm minimizes the ISL, and the MM-PMM algorithm utilizes numerous approximations to minimize the PSL, the two algorithms are more likely to have higher PSL than the proposed one, which minimizes the PSL with fewer approximations. Although few approximations in the proposed algorithm may result in a large amount of computation time, it is acceptable for real applications since the user can just produce multiple waveforms offline for the scenarios and store them in memory for real-time usage. 

Aside from the excellent performance of our algorithm, we notice that even the SCAN algorithm outperforms the MM-PMM algorithm in these cases. 
This phenomenon may result from our choice of parameters in the MM-PMM algorithm; however, it is difficult to find suitable parameters. The MM-PMM algorithm is so sensitive to the parameters that casual selection of parameters usually leads to divergence, not to mention producing a sequence better than the one obtained from the SCAN algorithm. Besides, it is also possible that the MM-PMM algorithm is less suitable for these cases since it was initially proposed to minimize the local PSL instead of the overall PSL \cite{LU2022}. 
On the other hand, the SCAN algorithm usually converges to a pretty good result even when we casually choose the parameters. Additionally, since its main goal is to minimize the overall ISL instead of the local ISL, it may be reasonable for the SCAN algorithm to outperform the MM-PMM algorithm in the considered cases.

\section{Performance Analysis by a New Lower Bound}\label{sec:PerformanceAnalysis}
Since Algorithm \ref{algo:AM} does not necessarily achieve the global optimality, further evaluation of the obtained solutions is needed. 
In this section, we attempt to derive a lower bound of the problem \eqref{prob:original} for the performance evaluation of our solutions shown in Section \ref{sec:NumericalValidation}. Commonly, solving the Lagrange dual problem is a good method to get a lower bound; however, a direct derivation of the Lagrange dual problem of the problem \eqref{prob:original} is difficult. Hence, in the following subsection, we propose a new technique to obtain a new lower bound of the problem \eqref{prob:original}. 
This new technique is based on the concept of the Lagrange dual problem with slight modification of the Lagrangian to circumvent the difficulty we would have faced in the direct derivation of the Lagrange dual problem.

\subsection{Lower Bound for PSL}
To make the derivation easier, we first reformulate the problem \eqref{prob:original} into the following equivalent problem (similar to the steps in Section \ref{sec:reformulation}).
\begin{subequations}\label{prob:Lower Bound for PSL}
\begin{align}
& \underset{\mathbf{x}\in\mathbb{C}^{N}, t\in \mathbb{R}}{\mathrm{minimize}} &&t \\
& \text{subject to} && \vert \mathbf{x}^H\mathbf{N}_{N}^\ell\mathbf{x} \vert\leq t, \forall \ell \in \mathbb{Z}_N\backslash \{0\} \\
&&& \mathbf{x}^H\mathbf{F}(f_i)\mathbf{x}\leq U_{\mathrm{max}}, \forall i\in\mathbb{N}_{N_f}\\
&&& \mathbf{x}^H\mathbf{E}_{n}^{(N)}\mathbf{x} =1, \forall n\in \mathbb{Z}_{N},
\end{align}
\end{subequations}
where $f_i$'s are $N_f$ points uniformly chosen from $\mathcal{F}_{\mathrm{\text{stop}}}$. 
Note that this problem is also equivalent to the problems \eqref{prob:MatrixRepresentationOriginal}, \eqref{prob:MatrixRepresentation}, and \eqref{prob:equivalent}, and the reasons are presented in Section \ref{sec:reformulation} along with these problems.

With the problem \eqref{prob:Lower Bound for PSL} in quadratic forms, we can start our derivation with its Lagrangian.
The Lagrangian $L: \mathbb{R}\times \mathbb{C}^{N}\times \mathbb{R}^{N-1}\times \mathbb{R}^{N_f}\times \mathbb{R}^{N} \rightarrow \mathbb{R}$ associated with the problem \eqref{prob:Lower Bound for PSL} is defined as  \cite{boyd_vandenberghe_2004}
\begin{equation*}
\begin{aligned}
L(t,\mathbf{x}, \bm{\lambda},\bm{\mu},\bm{\nu}) = t &+ 
\sum_{\ell=1}^{N-1}\lambda_\ell\left(\vert \mathbf{x}^H\mathbf{N}_{N}^\ell\mathbf{x} \vert - t\right)\\&+
\sum_{i=0}^{N_f-1}\mu_i \left(\mathbf{x}^H\mathbf{F}(f_i)\mathbf{x}-U_{\mathrm{max}}\right) \\
&+\sum_{n=0}^{N-1}\nu_n \left(\mathbf{x}^H\mathbf{E}_{n}^{(N)}\mathbf{x} -1 \right),
\end{aligned}
\end{equation*}
where $f_i$'s are $N_f$ points uniformly chosen from $\mathcal{F}_{\mathrm{\text{stop}}}$.
Due to the non-smoothness of the terms $\vert \mathbf{x}^H\mathbf{N}_{N}^\ell\mathbf{x} \vert$, it is difficult to directly derive the Lagrange dual function \cite{boyd_vandenberghe_2004} 
\begin{equation*}
g(\bm{\lambda},\bm{\mu},\bm{\nu})=\underset{t \in \mathbb{R}, \mathbf{x}\in \mathbb{C}^N}{\mathrm{inf}} ~L(t,\mathbf{x},\bm{\lambda},\bm{\mu},\bm{\nu})
\end{equation*}
for $\bm{\lambda}\succeq \bm{0}$, $\bm{\mu}\succeq \bm{0}$, and $\bm{\nu}\in\mathbb{R}^N$.
To address the issue, we augment the original Lagrangian with a new vector $\bm\theta \in \mathbb{R}^{N-1}$ and define the ``modified Lagrangian'' $\tilde{L}: \mathbb{R}\times \mathbb{C}^{N}\times \mathbb{R}^{N-1}\times \mathbb{R}^{N-1}\times \mathbb{R}^{N_f}\times \mathbb{R}^{N} \rightarrow \mathbb{R}$ as
\begin{equation}
\begin{aligned}
&\quad \tilde{L}(t,\mathbf{x},\bm\theta,\bm{\lambda},\bm{\mu},\bm{\nu})\\
&= t +
\sum_{\ell=1}^{N-1}\lambda_\ell\left(\mathbf{x}^H\left(\frac{ e^{-j\theta_\ell}\mathbf{N}_{N}^\ell+ e^{j\theta_\ell}{\mathbf{N}_{N}^\ell}^T}{2}\right)\mathbf{x} - t\right) \\
&\,\,\,\,~~~ +\sum_{i=0}^{N_f-1}\mu_i \left(\mathbf{x}^H\mathbf{F}(f_i)\mathbf{x}-U_{\mathrm{max}}\right)\\
&\,~~~~ +\sum_{n=0}^{N-1}\nu_n \left(\mathbf{x}^H\mathbf{E}_{n}^{(N)}\mathbf{x} -1 \right). 
\end{aligned}\label{eq:ModifiedLagrangian}
\end{equation}
 The next lemma characterizes the relation between the original Lagrangian and the modified one.
\begin{lem} \label{lemma:L>=}
For $\bm{\lambda} \succeq 0$ and $\bm\theta \in \mathbb{R}^{N-1}$, we always have 
\begin{equation*}
    L(t,\mathbf{x}, \bm{\lambda},\bm{\mu},\bm{\nu})\geq \tilde{L}(t,\mathbf{x},\bm\theta, \bm{\lambda},\bm{\mu},\bm{\nu}).
\end{equation*}
\end{lem}
\begin{proof}
Firstly, we express $\mathbf{x}^H\mathbf{N}_{N}^\ell\mathbf{x}$ as $a_\ell+b_\ell j$, where $j=\sqrt{-1}$ and $a_\ell,b_\ell\in\mathbb{R}$. Then, for any $c_\ell, d_\ell\in\mathbb{R}$, by the Cauchy-Schwartz inequality, we have 
\begin{equation*}
\begin{aligned}
\vert\mathbf{x}^H\mathbf{N}_{N}^\ell\mathbf{x} \vert\cdot\sqrt{c_\ell^2+d_\ell^2} &= \sqrt{a_\ell^2+b_\ell^2} \cdot\sqrt{c_\ell^2+d_\ell^2}\\
&\geq a_\ell c_\ell+b_\ell d_\ell\\
& = c_\ell \cdot \mathrm{Re}\{\mathbf{x}^H\mathbf{N}_{N}^\ell\mathbf{x}\} + d_\ell \cdot \mathrm{Im}\{\mathbf{x}^H\mathbf{N}_{N}^\ell\mathbf{x}\} \\
&= c_\ell \cdot \mathbf{x}^H\frac{\mathbf{N}_{N}^\ell+{\mathbf{N}_{N}^\ell}^T}{2}\mathbf{x} + d_\ell \cdot \mathbf{x}^H\frac{\mathbf{N}_{N}^\ell-{\mathbf{N}_{N}^\ell}^T}{2j}\mathbf{x} \\
& = \mathbf{x}^H\left(\frac{(c_\ell-d_\ell j)\mathbf{N}_{N}^\ell+(c_\ell+d_\ell j){\mathbf{N}_{N}^\ell}^T}{2}\right)\mathbf{x}.
\end{aligned}
\end{equation*}
Therefore,
\begin{equation}
\begin{aligned}
\vert\mathbf{x}^H\mathbf{N}_{N}^\ell\mathbf{x} \vert & \geq \mathbf{x}^H\left(\frac{(c_\ell-d_\ell j)\mathbf{N}_{N}^\ell+(c_\ell+d_\ell j){\mathbf{N}_{N}^\ell}^T}{2\sqrt{c_\ell^2+d_\ell^2}}\right)\mathbf{x}\\
&= \mathbf{x}^H\left(\frac{(A_\ell e^{-j\theta_\ell})\mathbf{N}_{N}^\ell+(A_\ell e^{j\theta_\ell}){\mathbf{N}_{N}^\ell}^T}{2A_\ell}\right)\mathbf{x}\\
&= \mathbf{x}^H\left(\frac{ e^{-j\theta_\ell}\mathbf{N}_{N}^\ell+ e^{j\theta_\ell}{\mathbf{N}_{N}^\ell}^T}{2}\right)\mathbf{x},
\end{aligned}\label{eq:ineq_modeified_Lagrangian}
\end{equation}
where $A_\ell e^{j\theta_\ell}$ is the polar form of $c_\ell+ d_\ell j $. By replacing the terms $\vert\mathbf{x}^H\mathbf{N}_{N}^\ell\mathbf{x} \vert$ in $L(t,\mathbf{x},\bm{\lambda},\bm{\mu},\bm{\nu})$ with $ \mathbf{x}^H\left(\frac{ e^{-j\theta_\ell}\mathbf{N}_{N}^\ell+ e^{j\theta_\ell}{\mathbf{N}_{N}^\ell}^T}{2}\right)\mathbf{x},$ we will obtain $\tilde{L}(t,\mathbf{x}, \bm\theta,\bm{\lambda},\bm{\mu},\bm{\nu})$. Then, due to the assumption $\bm\lambda \succeq 0$ and inequality \eqref{eq:ineq_modeified_Lagrangian}, we have 
$$L(t,\mathbf{x},\bm{\lambda},\bm{\mu},\bm{\nu})\geq \tilde{L}(t,\mathbf{x}, \bm\theta,\bm{\lambda},\bm{\mu},\bm{\nu}).$$ 
In addition, the equality can be achieved by some proper choice of $\bm\theta$.
\end{proof}

With the important inequality provided in Lemma \ref{lemma:L>=}, we will be able to derive a problem with its optimal value serving as an lower bound of the problem \eqref{prob:Lower Bound for PSL} later.
For the convenience of further derivation, we reformulate the modified Lagrangian in \eqref{eq:ModifiedLagrangian} as
\begin{align*}
&\quad \tilde{L}(t,\mathbf{x},\bm\theta, \bm{\lambda},\bm{\mu},\bm{\nu})\\
&= t \left(1-\sum_{\ell=1}^{N-1}\lambda_\ell\right)  -\sum_{i=1}^{N_f}\mu_i U_{\mathrm{max}} -\sum_{n=0}^{N-1}\nu_n+\mathbf{x}^H \mathbf{M}_{\bm{\theta}}(\bm{\lambda},\bm{\mu},\bm{\nu}) \mathbf{x},
\end{align*}
where 
\begin{align}
\mathbf{M}_{\bm{\theta}}(\bm{\lambda},\bm{\mu},\bm{\nu})&=\sum_{\ell=1}^{N-1}\frac{\lambda_\ell e^{-j\theta_\ell}\mathbf{N}_{N}^\ell+ \lambda_\ell e^{j\theta_\ell}{\mathbf{N}_{N}^\ell}^T}{2}\nonumber \\
&\quad + \sum_{i=0}^{N_f-1}\mu_i\mathbf{F}(f_i)+ \sum_{n=0}^{N-1}\nu_n\mathbf{E}_{n}^{(N)}.\label{eq:M_matrix}
\end{align}
Inspired by the relation of Lagrangian and the dual function presented in \cite{boyd_vandenberghe_2004}, we define a ``modified dual function'' via $\tilde{L}$ for each $\bm\theta \in \mathbb{R}^{N-1}$ as follows:
\begin{align}
\tilde{g}_{\bm\theta}(\bm{\lambda},\bm{\mu},\bm{\nu}) 
&=\underset{t \in \mathbb{R}, \mathbf{x}\in \mathbb{C}^N}{\mathrm{inf}} ~\tilde{L}(t,\mathbf{x},\bm\theta,\bm{\lambda},\bm{\mu},\bm{\nu})\nonumber\\
& =  - U_{\mathrm{max}}\sum_{i=0}^{N_f-1}\mu_i - \sum\limits_{n=0}^{N-1}\nu_n, \label{eq:g_tilde_theta_lambda_mu_nu}
\end{align}
with 
\begin{align}
&\quad\boldsymbol{\mathbf{dom}}~\tilde{g}_{\bm\theta}\nonumber \\
&=\Bigg\{ ( \bm{\lambda},\bm{\mu},\bm{\nu}) \Bigg\vert \bm\lambda \succeq 0,\bm\mu \succeq0, \sum\limits_{\ell=1}^{N-1}\lambda_\ell=1, \mathbf{M}_{\bm{\theta}}(\bm{\lambda},\bm{\mu},\bm{\nu})\succeq 0\Bigg\},\label{eq:dom_g_tilde_theta}
\end{align}
where $\mathbf{M}_{\bm{\theta}}(\bm{\lambda},\bm{\mu},\bm{\nu})$ was defined in \eqref{eq:M_matrix}.
Then, we have the following theorem.
\begin{thm}\label{thm:DualFunctionLowerBound}
Suppose $\bm\theta \in \mathbb{R}^{N-1}$. Then, for any feasible point $(t', \mathbf{x}')$ of the problem \eqref{prob:Lower Bound for PSL}, we always have
\begin{equation}\label{ineq:Lower Bound for PSL}
\tilde{g}_{\bm\theta}(\bm{\lambda},\bm{\mu},\bm{\nu}) \leq t'
,\forall (\bm{\lambda},\bm{\mu},\bm{\nu})\in \boldsymbol{\mathbf{dom}}~\tilde{g}_{\bm\theta}.
\end{equation}

\end{thm}
\begin{proof}
$$
\begin{aligned}
\tilde{g}_{\bm\theta}( \bm{\lambda},\bm{\mu},\bm{\nu}) &= \underset{t \in \mathbb{R}, \mathbf{x}\in \mathbb{C}^N}{\mathrm{inf}} ~\tilde{L}(t,\mathbf{x},\bm\theta, \bm{\lambda},\bm{\mu},\bm{\nu})\\
&\leq \tilde{L}(t',\mathbf{x}',\bm\theta, \bm{\lambda},\bm{\mu},\bm{\nu}) \\
&\leq L(t',\mathbf{x}', \bm{\lambda},\bm{\mu},\bm{\nu})\\
& \leq t'.
\end{aligned} 
$$
The second inequality follows from Lemma \ref{lemma:L>=}. The third inequality is the result of $\bm\lambda \succeq 0,~\bm\mu \succeq0$, and the assumption that $(t', \mathbf{x}')$ is a feasible point of the problem \eqref{prob:Lower Bound for PSL}. 
\end{proof}
With Theorem \ref{thm:DualFunctionLowerBound}, we know that the modified dual function can always provide a lower bound for the problem \eqref{prob:Lower Bound for PSL} no matter which $\bm\theta$ is chosen. Then, in order to obtain the largest lower bound, we firstly demonstrate the following corollary.
\begin{cor}\label{cor:dual_function}
Suppose $t^\star$ is the optimal value of the problem \eqref{prob:Lower Bound for PSL}. Then, \begin{equation*}
   \underset{\begin{subarray}{c}\bm\theta \in \mathbb{R}^{N-1},\\ (\bm\lambda,\bm\mu,\bm\nu)  \in \boldsymbol{\mathbf{dom}}~ \tilde{g}_{\bm\theta}\end{subarray}}{\mathrm{sup}}
   \tilde{g}_{\bm\theta}( \bm\lambda,\bm\mu,\bm\nu) \leq t^\star.
    \end{equation*}
\end{cor}
\begin{proof}
The result can be directly derived from $\textbf{Theorem \ref{thm:DualFunctionLowerBound}}$ by taking supremum over $(\bm\theta,\bm\lambda,\bm\mu,\bm\nu)$ on the left side of \eqref{ineq:Lower Bound for PSL}.
\end{proof}
Then, we can rewrite $\textbf{Corollary \ref{cor:dual_function}}$ as an optimization problem as follows.
\begin{subequations}\label{prob: Largest Lower Bound for PSL}
\begin{align}
&\underset{\begin{subarray}{c}\bm\theta \in\mathbb{R}^{N-1},\\
(\bm\lambda,\bm\mu,\bm\nu) \in  \mathbb{R}^{N-1}\times\mathbb{R}^{N_{f}}\times\mathbb{R}^{N}\end{subarray}}{\mathrm{maximize}} &&  - U_{\mathrm{max}}\sum_{i=0}^{N_f-1}\mu_i - \sum_{n=0}^{N-1}\nu_n \\
&~~~~~~\text{subject to}
&& \mathbf{M}_{\bm{\theta}}(\bm{\lambda},\bm{\mu},\bm{\nu})\succeq 0\\
&&& \sum_{\ell=1}^{N-1}\lambda_\ell=1\\
&&& \bm\lambda \succeq 0,~\bm\mu \succeq 0,
\end{align}
\end{subequations}
where $\mathbf{M}_{\bm{\theta}}(\bm{\lambda},\bm{\mu},\bm{\nu})$ was defined in \eqref{eq:M_matrix}.
Since $\lambda_\ell \geq 0$, $\lambda_\ell$ and $\theta_\ell$ can be combined further as $\lambda_\ell e^{j\theta_\ell}$, which enables us to substitute a new complex variable $y_\ell$ for $\lambda_\ell e^{j\theta_\ell}$. Therefore, the problem \eqref{prob: Largest Lower Bound for PSL} can be reformulated as follows.
\begin{subequations}\label{prob: Combined Largest Lower Bound for PSL}
\begin{align}
& \underset{\begin{subarray}{c}\mathbf{y}\in\mathbb{C}^{N-1},\\(\bm\mu,\bm\nu) \in  \mathbb{R}^{N_{f}}\times\mathbb{R}^{N}\end{subarray}}{\mathrm{maximize}} && - U_{\mathrm{max}}\sum_{i=0}^{N_f-1}\mu_i - \sum_{n=0}^{N-1}\nu_n \\
&~~ \text{subject to} && \left\|\mathbf{y}\right\|_1=1\\
&&& \mathbf{M}(\mathbf{y},\bm\mu,\bm\nu)\succeq 0\\
&&& \bm\mu\succeq 0,
\end{align}
\end{subequations}
where $\mathbf{M}(\mathbf{y},\bm\mu,\bm\nu)$ is defined as
\begin{align}
\mathbf{M}(\mathbf{y},\bm\mu,\bm\nu)=&\sum_{\ell=1}^{N-1}\frac{y_\ell^{*}\mathbf{N}_{N}^\ell+ y_\ell{\mathbf{N}_{N}^\ell}^T}{2}+ \sum_{i=0}^{N_f-1}\mu_i\mathbf{F}(f_i)+ \sum_{n=0}^{N-1}\nu_n\mathbf{E}_{n}^{(N)}. \label{eq: M_y matrix}
\end{align}

Although both the problems \eqref{prob: Largest Lower Bound for PSL} and \eqref{prob: Combined Largest Lower Bound for PSL} are not convex optimization problems due to $\bm\theta$ in the problem \eqref{prob: Largest Lower Bound for PSL} and the constraint $\left\|\mathbf{y}\right\|_1=1$ in the problem \eqref{prob: Combined Largest Lower Bound for PSL}, we can relax the constraint $\left\|\mathbf{y}\right\|_1=1$ in the problem \eqref{prob: Combined Largest Lower Bound for PSL} as $\left\|\mathbf{y}\right\|_1\leq 1$ without losing any information about the lower bound (The explanation will be provided later). 

Therefore, instead of solving the problem \eqref{prob: Largest Lower Bound for PSL} 
or the problem \eqref{prob: Combined Largest Lower Bound for PSL}, we can solve the following convex optimization problem to obtain a lower bound for PSL.
\begin{subequations}\label{prob: Relaxed Largest Lower Bound for PSL}
\begin{align}
& \underset{\begin{subarray}{c}\mathbf{y}\in\mathbb{C}^{N-1},\\(\bm\mu,\bm\nu) \in  \mathbb{R}^{N_{f}}\times\mathbb{R}^{N}\end{subarray}}{\mathrm{maximize}} && - U_{\mathrm{max}}\sum_{i=0}^{N_f-1}\mu_i - \sum_{n=0}^{N-1}\nu_n \label{prob: Relaxed Largest Lower Bound for PSL (a)}\\
&~~ \text{subject to} && \left\|\mathbf{y}\right\|_1\leq 1\\
&&& \mathbf{M}(\mathbf{y},\bm\mu,\bm\nu)\succeq 0\\
&&& \bm\mu\succeq 0,
\end{align}
\end{subequations}
where $ \mathbf{M}(\mathbf{y},\bm\mu,\bm\nu) $ is defined in \eqref{eq: M_y matrix}. 
The reason why solving the problem \eqref{prob: Relaxed Largest Lower Bound for PSL} instead of the problem \eqref{prob: Combined Largest Lower Bound for PSL} is legitimate can be seen in the following theorem since a non-positive lower bound for PSL is always meaningless.
\begin{thm}\label{thm:relax_y}
The optimal value of the problem \eqref{prob: Relaxed Largest Lower Bound for PSL} is the same as that of the problem \eqref{prob: Combined Largest Lower Bound for PSL} whenever the optimal value of the problem \eqref{prob: Relaxed Largest Lower Bound for PSL} is not zero.
\end{thm}
\begin{proof}
The proof of Theorem \ref{thm:relax_y} is given in Appendix \ref{appendix:proof_of_thm_relax_y}.
\end{proof}

Finally, due to the convexity of the problem \eqref{prob: Relaxed Largest Lower Bound for PSL}, it can be quickly solved via CVX and serve as an estimate of the duality gap \cite{boyd_vandenberghe_2004}.

\begin{rmk}
When the optimal value of the problem \eqref{prob: Relaxed Largest Lower Bound for PSL} is less than 1, that of the problem \eqref{prob: Combined Largest Lower Bound for PSL} will also be less than 1 since the feasible set of the problem \eqref{prob: Combined Largest Lower Bound for PSL} is a subset of that of the problem \eqref{prob: Relaxed Largest Lower Bound for PSL}. In this case, these two problems do not provide information on PSL since PSL is at least 1 because we always have $\vert r_{N-1}\vert = \vert r_{-N+1}\vert=1$.
\end{rmk}
Note that, based on our experimental experience, an optimal value not larger than one rarely occurs when the spectral constraints are properly set with $N_f\neq 0$ and $U_{\mathrm{max}}<N$. Therefore, the assumption of Theorem \ref{thm:relax_y} are generally true and the resulting lower bound from the problem \eqref{prob: Relaxed Largest Lower Bound for PSL} can usually be used to evaluate the designed waveform.

\subsection{Numerical Results}
Solving the problem \eqref{prob: Relaxed Largest Lower Bound for PSL} for different $N$, $U_{\mathrm{max}}$, $~\mathcal{F}_{\text{stop}}$, and $N_f$, we can obtain lower bounds for the optimal NPSL under different spectral constraints. 
These lower bounds can be used to approximate the distance between the optimal NPSL and the NPSLs of all the waveforms obtained from our algorithm.
We summarize the lower bounds of the cases previously presented in Section \ref{Sec:Numerical_Examples} in
Table \ref{LB-total}.
From this table, it can be seen that our algorithm attains solutions with duality gap less than 3 dB for the cases where $\mathcal{F}_{\mathrm{stop}}$ is $[0.2,0.3]$, which demonstrates the fact that the PSLs of these sequences we obtained are quite close to the optimal solution. 
Apart from the optimality evaluation in the PSL minimization problem, this theoretical lower bound can also provide all the other problems in similar forms with an alternative lower bound when their Lagrange dual problems are difficult to derive or solve. Therefore, aside from waveform design for active sensing systems, this proposed theory in lower bound may even be useful in a variety of applications.

\begin{table}[htbp]
\centering
\setlength\tabcolsep{3.85pt}
    \begin{threeparttable}[htbp]
        \caption{NPSL of the Numerical Examples and the Corresponding Lower Bounds (in decibel)}
        \label{LB-total}
        \begin{tabular*}{\columnwidth}{c c c c c c}
          \toprule
          \textbf{Case} &\textbf{SCAN} &\textbf{MM-PMM} & \makecell{\textbf{Proposed}\\ \textbf{Algorithm}}  & \makecell{\textbf{Proposed} \\\textbf{Lower Bound}}& \makecell{\textbf{Lower bound}\\\textbf{in} \cite{McCormick2017}} \\
          \midrule
          1 & -16.52&-14.68  & -18.18 & -20.27& -30.10\\
          2 & -17.91 &-16.20  & -21.16 & -22.92& -40.00\\
          3 & -25.08 &-20.86  & -26.88 & -32.00& -40.00\\
          4 & -25.56 &-24.58  & -28.69 & -32.86& -42.14\\
          5 & -18.72 &-18.49  & -22.40 & -23.95& -48.16\\
          \bottomrule
        \end{tabular*}
    \end{threeparttable}
\end{table}

\section{Discussions With Related Works}\label{sec:discussion}
In this section, we raise some discussions on the relationship of the problem studied in this paper and some broader class of waveform design problems in the fields of cognitive radars and active sensors.
Specifically, we will first elaborate on sequence design methods without special compatibility, and then on waveform design methods that seek to shape the waveform's ambiguity function \cite{Sussman1962, he_li_stoica_2012, Wolf1969} according to some application requirements.
These discussions intend to provide with a more comprehensive view that connects the contribution of this work and the broader literature on similar topics.
\subsection{Sequences Without Spectral Compatibility}
As mentioned in the introduction, many works \cite{He2009, Mohammad2017, Song2016April, Song2016June, Frank1963, Golomb1965, Hamid2017, Raei2022} have aimed at designing sequences without spectral compatibility. 
Without spectral constraints, they are usually expected to achieve a better auto-correlation than those with such constraints.
For example, reference \cite{Hamid2017} also studies methods to minimize PSL of the sequence's auto-correlation function, yet without a spectral constraint.
Here, we compare the sequence generated by the proposed algorithm in Case 2 in Section \ref{sec:case2} with the sequence generated by the PSL Optimization Cyclic Algorithm (POCA) in \cite{Hamid2017} and the Golomb sequence \cite{Golomb1993}. 
The results are shown in Fig. \ref{figN100F2030_no_spectral} and Table \ref{tableN100F2030_no_spectral}, with the former containing the normalized auto-correlations and energy spectra of all sequences in comparison, and the latter summarizing their normalized peak side-lobe levels and stopband attenuations. 
In the simulation, the POCA is initialized by the Golomb sequence and conducted with\footnote{In the POCA implemented here, the sequence is projected to unit circle for every iteration (i.e., the intermediate step for unimodularity constraint mentioned in \cite{Hamid2017}).} $N=100$, $Q=100$, and $\varepsilon=10^{-4}$. 
The result confirms the intuition that sequences with only the unimodular constraint 
have more freedom to achieve a lower PSL than the proposed one which 
was limited
by spectral constraints. 
We can therefore say that the proposed algorithm has sacrificed the PSL performance by a loss of around $8$ to $9$ dB in exchange for the capability of spectral compatibility, which is a critical feature in applications where the coexistence of multiple radios is demanded. 

\begin{table}[t]
  \centering
  \begin{threeparttable}
    \caption{NPSL and $A_{\mathrm{stop}}$ for different methods.}
    \label{tableN100F2030_no_spectral}
    \begin{tabular}{c c c}
      \toprule
      \textbf{Method} & \textbf{NPSL} (dB) & \textbf{{}Stopband  Attenuation} (dB) \\
      \midrule
      Golomb \cite{Golomb1993} & -26.32 & -1.25 \\
       POCA \cite{Hamid2017} & -29.99 &  -1.21 \\
      Proposed Algorithm &  -21.16 & 30.14 \\
      \bottomrule
    \end{tabular}
  \end{threeparttable}
\end{table}
\begin{figure}[t]
    \centering
    \includegraphics[width=0.48\textwidth]{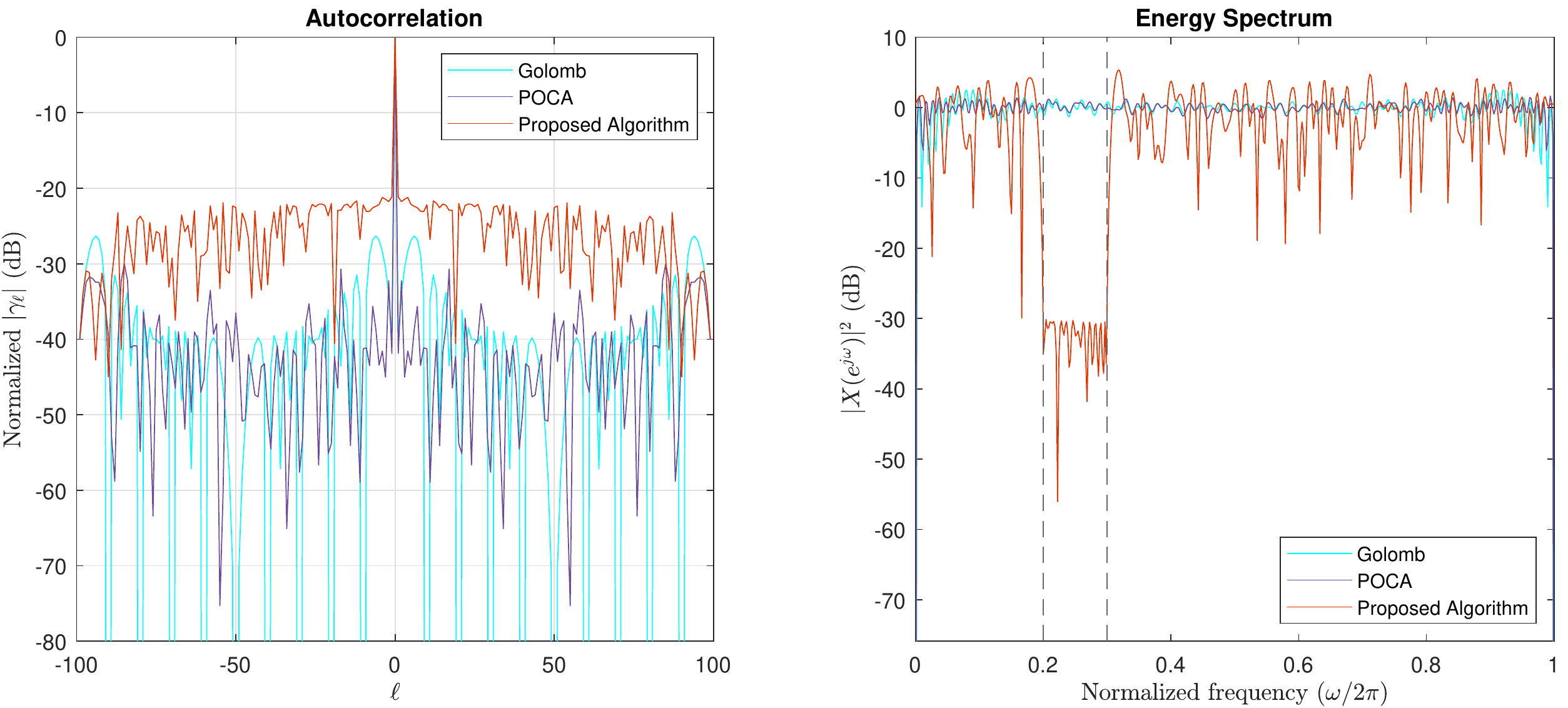}
    \caption{The auto-correlations and energy spectra comparison of the proposed one in Case 2 and other sequences without spectral compatibility}
    \label{figN100F2030_no_spectral}
\end{figure}

\subsection{Relations to Ambiguity Functions}
There are also a number of waveform design methods \cite{Feng2017, Latham2019, Mohammad2019, Wang2021, Cui2017_AF, Yang2018, Yang2020} in the recent literature that concern, in addition to behaviors of the sequences' auto-correlation functions, their ambiguity functions (AF).
As a two-variable function, an AF measures variations of the waveform's matched filter responses along both the range and Doppler domains and can provide much more information about the waveform's characteristics, particularly in the Doppler domain, compared to the single-variable auto-correlation function.
In fact, the auto-correlation function is just the zero-Doppler cut of an AF \cite{he_li_stoica_2012}.
This enables the AF-shaping waveform design methods to control the waveform's Doppler-domain requirements more directly, in order to adapt to various applications whose requirements may involve large relative speeds between targets and sensing systems.

Waveform design methods that consider AF-shaping may be roughly categorized into two purposes \cite[Ch. 6]{he_li_stoica_2012}: one aims to create \emph{Doppler-resilient} waveforms and the other focuses on generating \emph{Doppler-sensitive} waveforms.
Methods that focus on creating a Doppler-resilient waveform are to maintain good correlation properties even in the presence of a nonzero Doppler shift \cite{Feng2017, Latham2019, Mohammad2019, Wang2021}, at the expense of an error in the range estimation of a target. 
On the other hand, AF-shaping methods focusing on Doppler-sensitive waveforms generally
suppress the local side-lobes of the AF of the sequence, i.e., in both range and Doppler domains.
These methods \cite{Cui2017_AF, Yang2018, Yang2020} regulate the interference power from the unwanted return in certain range-Doppler bins, and generally produce waveforms whose AFs possess a so-called \emph{thumbtack-shaped} property\cite{he_li_stoica_2012}. 

In comparison, this paper, as well as many aforementioned works \cite{He2009, Mohammad2017, Song2016April, Song2016June, Frank1963, Golomb1965, Hamid2017, Raei2022, Cui2018, FAN2021107960, LU2022}, focuses \emph{only} on the shaping of auto-correlation, i.e., the zero-Doppler cut of AF. 
The shaping of auto-correlation without considering the whole AF might result in some unwanted properties in the presence of the Doppler effect. 
Nevertheless, when we plot the AF of the designed waveforms in this paper, we found that they also happen to resemble a thumbtack shape.
For example, in Fig. \ref{figAF_linear}, the discrete AF of the sequence generated by Algorithm \ref{algo:AM} in Case 2 in Section \ref{sec:case2} is shown\footnote{For conciseness, the readers are referred to \cite{he_li_stoica_2012, Cui2017_AF, Yang2020} for the exact definition of AF used for generating Fig. \ref{figAF_linear}.}. 
Even though we did not impose any constraints against the non-zero Doppler side-lobes in Algorithm \ref{algo:AM}, the magnitudes of responses on non-zero Doppler cuts of the AF does not significantly increase\footnote{Yet, the largest peak side-lobe of the AF in Fig. \ref{figAF_linear} is $-10.79$ dB, still considerably larger than the PSL in its auto-correlation ($-21.16$ dB).}, at least for normalized Doppler frequencies that are close to zero. 
One speculation about the reason is that it is initialized by the SCAN-generated sequences\cite{He2010}, which also have the thumbtack-like AFs.
Such a property may already make the waveform applicable in a sensing system that needs to detect moderate-speed targets.\footnote{Take an example excerpted from \cite{he_li_stoica_2012}: consider an X-band radar whose operating wavelength is $\lambda = 3$cm and a baseband sampling rate of $100$ MHz. Then, a fighter jet moving at speed $1020$ m/s (i.e., roughly Mach 3) would induce a Doppler frequency of only $f = 2v/\lambda = (2\times 1020)\mathrm{(m/s)}/0.03\mathrm{m} = 68\mathrm{kHz}$ \cite{he_li_stoica_2012}. If the sampling rate is $100$-MHz, then its normalized Doppler frequency\cite{Cui2017_AF,Yang2020}, defined as $2vT_s/\lambda$, is even smaller than $10^{-3}$.
}

According to the previous discussions, it is observed that when the application considers relatively high normalized Doppler frequency \cite{Cui2017_AF, Yang2018, Yang2020}, (e.g., very high-speed targets, using short wavelengths, or in a narrow-band scenario,  \cite{Yang2018}), shaping the AF becomes critical for probing sequence design. 
On the other hand, when the application scenarios concern only targets of moderate speed or the maximum possible target speed in the application induces just a tiny ``normalized Doppler frequency" (i.e., longer wavelengths are used, wide-band), a problem simpler than AF-shaping, such as considering only the auto-correlation function 
as in this paper, could already be sufficient in most applications that require detecting targets with moderate speed.


\begin{figure}[t]
    \centering
    \includegraphics[width=0.48\textwidth]{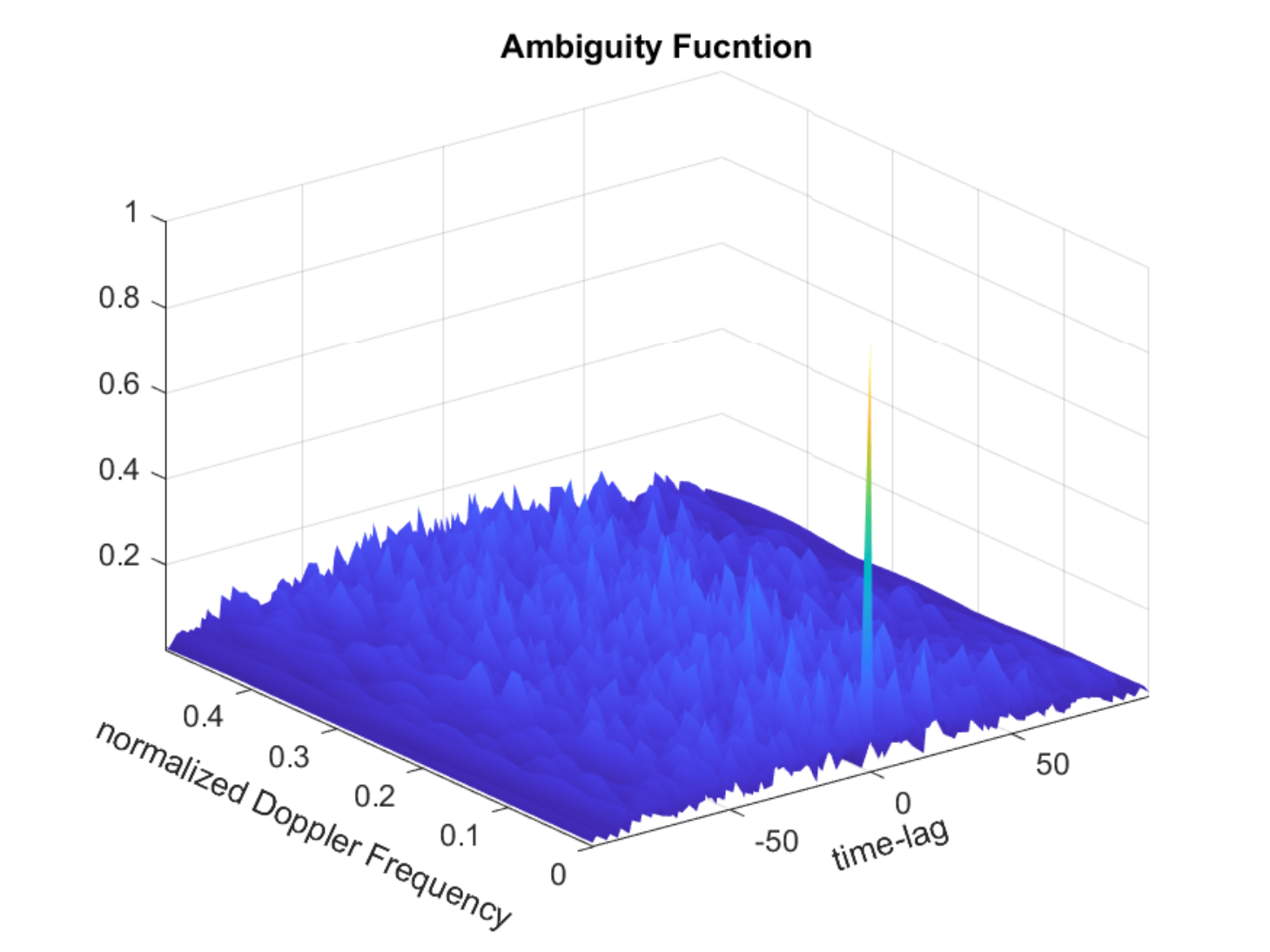}
    \caption{The ambiguity function of the sequence generated by the proposed algorithm in Case 2}
    \label{figAF_linear}
\end{figure}

\section{Conclusion}\label{sec:conclusion}
We propose a new algorithm via alternating minimization for the design of unimodular sequences with controllable spectral energy in predetermined stopbands and an approximately optimal PSL. 
Since the stopband requirements are not included in the objective function, the proposed method has more flexibility in spectral adjustment compared to other algorithms. Numerical results in Section \ref{sec:NumericalValidation} demonstrate the advantages of the proposed method both in the PSL and the spectral compatibility over SCAN. 
In addition, we also derived a lower bound for the PSL from its Lagrangian to evaluate the duality gap between the optimal value and the attained one. 
The key ideas of the derivation are the introduction of new variables and the replacement of non-smooth terms via Cauchy-Schwartz inequality. 
These skills are useful in obtaining a lower bound and can be applied to different kinds of problems, especially those associated with the absolute values of some complex quadratic forms of non-Hermitian matrices.
The numerical results of the derived lower bound for the PSL show that the proposed method has the potential of achieving a near-optimal solution due to the narrow duality gaps. 
In the future, a narrower duality gap may be able to be achieved by choosing different parameters in Algorithm \ref{algo:AM} or different penalty functions in the problem \eqref{prob:AM_first} and \eqref{prob:AM_second} to attain a lower PSL, or by finding a tighter lower bound. Other cases like local PSL minimization or considering different spectral masks may be able to be implemented by the proposed scheme via choosing a desired set for $\ell$ in the PSL constraint instead of $\mathbb{Z}_N\backslash \{0\}$ or via selecting different $U_{\mathrm{max}}$, say,  $U_{\mathrm{max},i}$,  for each $f_i$ in the stopband, respectively. 

%

\appendices

\section{Proof of Theorem \ref{thm:TwoMatrixIneq}}\label{appendix:proof_of_thm_TwoMatrixIneq}
Consider the real vector space $\mathbb{H}^n$ with the inner product defined by $\langle \mathbf{A},\mathbf{B}\rangle=\mathrm{tr}(\mathbf{AB})$. 
By Cauchy–Schwarz inequality, we have
\begin{equation} \label{eq:Cauchy-Schwarz}
\mathrm{tr}(\mathbf{AB})\leq \sqrt{\mathrm{tr(\mathbf{A^2})}}\sqrt{\mathrm{tr(\mathbf{B^2})}}=\sqrt{\sum_i\lambda_{i,a}^2}\sqrt{\sum_j\lambda_{j,b}^2},
\end{equation}
where $\lambda_{i,a}$ and $\lambda_{j,b}$ are the eigenvalues of $A$ and $B$, respectively. Since $\mathbf{A}$ and $\mathbf{B}$ are both positive semidefinite, their eigenvalues are all nonnegative. Hence, 
\begin{equation} \label{eq:ineq_for_rank_one_condition}
\sqrt{\sum_i\lambda_{i,a}^2}\sqrt{\sum_j\lambda_{j,b}^2}\leq \sum_i\lambda_{i,a}\sum_j\lambda_{j,b}=\mathrm{tr}(\mathbf{A})\mathrm{tr}(\mathbf{B}),
\end{equation}
where the inequality can be obvious by taking the square of both sides. 
When the equality of (\ref{eq:Cauchy-Schwarz}) holds, Cauchy–Schwarz inequality gives us that $\mathbf{A}$ and $\mathbf{B}$ are linearly dependent. 
In addition, when the equality of (\ref{eq:ineq_for_rank_one_condition}) holds, since 
\begin{equation}
\sqrt{\sum_i\lambda_{i,a}^2}=\sum_i\lambda_{i,a},\sqrt{\sum_j\lambda_{j,b}^2}=\sum_i\lambda_{j,b},
\end{equation}
and all the eigenvalues are non-negative, $\mathbf{A}$ and $\mathbf{B}$ are of rank at most one.
\section{Proof of Theorem \ref{thm:relax_y}}\label{appendix:proof_of_thm_relax_y}
In the problem \eqref{prob: Relaxed Largest Lower Bound for PSL}, since $(\mathbf{y},\bm\mu,\bm\nu)=(0,0,0)$ is always a feasible point, which results in a zero objective function value, we always have the optimal value being larger than or equal to zero.
Then, suppose the optimal value of the problem \eqref{prob: Relaxed Largest Lower Bound for PSL} is not zero, i.e., $d^\star=-U_{\mathrm{max}}\sum_{f_i\in\mathcal{F}_{\mathrm{\text{stop}}}}\mu_i^\star - \sum_{n=0}^{N-1}\nu_n^\star > 0$. We prove that $\left\|\mathbf{y}^\star\right\|_1=1$. Firstly, we assume that $\left\|\mathbf{y}^\star\right\|_1=\eta$ and $0<\eta<1$. By taking $\mathbf{y}^\prime =  \frac{1}{\eta}\mathbf{y}^\star,\bm\mu^\prime=\frac{1}{\eta}\bm\mu^\star, \bm\nu^\prime=\frac{1}{\eta}\bm\nu^\star$, we have 
$$
\begin{aligned}
\left\|\mathbf{y}^\prime\right\|_1=1,~\bm\mu^\prime\succeq 0,
\end{aligned}
$$
and
$$
\begin{aligned}
&\sum_{\ell=1}^{N-1}\frac{(y_\ell^\prime)^{*}\mathbf{N}_{N}^\ell+ y_\ell^\prime{\mathbf{N}_{N}^\ell}^T}{2} + \sum_{i=0}^{N_f-1}\mu_i^\prime\mathbf{F}(f_i)+ \sum_{n=0}^{N-1}\nu_n^\prime\mathbf{E}_{n}^{(N)}\\
&=\frac{1}{\eta}\left(\sum_{\ell=1}^{N-1}\frac{(y_\ell^\star)^{*}\mathbf{N}_{N}^\ell+ y_\ell^\star{\mathbf{N}_{N}^\ell}^T}{2} + \sum_{i=0}^{N_f-1}\mu_i^\star\mathbf{F}(f_i)+ \sum_{n=0}^{N-1}\nu_n^\star\mathbf{E}_{n}^{(N)}\right)\\
&\succeq \frac{1}{\eta}0 \\
&=0,
\end{aligned}
$$
which implies that $(\mathbf{y}^\prime,\bm\mu^\prime,\bm\nu^\prime)$ is a feasible point of the problem \eqref{prob: Relaxed Largest Lower Bound for PSL}.
However,
\begin{align*}
d^\prime &=- U_{\mathrm{max}}\sum_{i=0}^{N_f-1}\mu_i^\prime - \sum_{n=0}^{N-1}\nu_n^\prime\\
& = \frac{1}{\eta}\left(- U_{\mathrm{max}}\sum_{f_i\in\mathcal{F}_{\mathrm{\text{stop}}}}\mu_i^\star - \sum_{n=0}^{N-1}\nu_n^\star \right)\\
&= \frac{1}{\eta}d^\star \\
& > d^\star,
\end{align*}
which contradicts the assumption of optimality. Therefore, $\left\|\mathbf{y}\right\|_1 \notin (0,1)$.
Secondly, we assume $\left\|\mathbf{y}^\star\right\|_1=0$, which implies that $\mathbf{y}^\star=0$. Then, by taking any $\xi>1$ and setting $\bm\mu^\prime=\xi\bm\mu^\star,~ \bm\nu^\prime=\xi\bm\nu^\star$, we have
\begin{align*}
& \sum_{i=0}^{N_f-1}\mu_i^\prime\mathbf{F}(f_i)+ \sum_{n=0}^{N-1}\nu_n^\prime\mathbf{E}_{n}^{(N)}\\
&=\xi\left(\sum_{i=0}^{N_f-1}\mu_i^\star\mathbf{F}(f_i)+ \sum_{n=0}^{N-1}\nu_n^\star\mathbf{E}_{n}^{(N)}\right)\\
&\succeq \xi0 \\
&=0,
\end{align*}
which implies that $(0,\bm\mu^\prime,\bm\nu^\prime)$ is a feasible point of the problem \eqref{prob: Relaxed Largest Lower Bound for PSL}.
Nevertheless,

\begin{align*}
d^\prime &=- U_{\mathrm{max}}\sum_{i=0}^{N_f-1}\mu_i^\prime - \sum_{n=0}^{N-1}\nu_n^\prime\\
& = \xi\left(- U_{\mathrm{max}}\sum_{i=0}^{N_f-1}\mu_i^\star - \sum_{n=0}^{N-1}\nu_n^\star \right)\\
&= \xi d^\star \\
& > d^\star,
\end{align*}
which contradicts the assumption of optimality. As a result, we prove that $\left\|\mathbf{y}\right\|_1 = 1$ whenever the optimality is attained with the optimal value $d^\star$ being positive.
\section*{Acknowledgment}
\noindent This work was supported by the Ministry of Science and Technology of Taiwan under Grant MOST 110-2221-E-002-074.

\ifCLASSOPTIONcaptionsoff
  \newpage
\fi



\renewcommand*{\bibfont}{\footnotesize}
\bibliographystyle{IEEEtranN}
\bibliography{References}

%

%











\end{document}